\documentclass{article}
\usepackage{graphicx}
\usepackage{hyperref}
\usepackage{amsmath,amsfonts,amssymb,amsthm}
\usepackage[dvipsnames]{xcolor}

\hypersetup{
colorlinks=true,
urlcolor=blue,
linkcolor=blue,
citecolor=OliveGreen,
}

\usepackage{bm}
\usepackage{csquotes}
\allowdisplaybreaks

\usepackage{cleveref}
\usepackage{enumerate}
\usepackage[sortcites, style=alphabetic, maxbibnames=99, hyperref]{biblatex}
\usepackage{geometry}
\usepackage{xcolor}
\usepackage{algorithm}
\usepackage{algpseudocode}
\usepackage{setspace}
\usepackage{bbm}
\usepackage{thm-restate}
\usepackage{mathabx}
\usepackage{thmtools}
\usepackage{csquotes}
\usepackage{authblk}
\usepackage{bm}
\usepackage{tikz}
\usepackage{tikz-cd}

\global\long\def\emd{\mathrm{EMD}}
\global\long\def\tv{\mathrm{TV}}

\newcommand{\findclosepairs}{\texttt{FindClosePairs}}

\newcommand{\sscp}{\texttt{SubsCP}}
\newcommand{\llp}{\texttt{LastSmallPrefix}}

\newcommand{\constsampler}{\texttt{ConstantSampler}}
\newcommand{\arbitsample}{\texttt{ArbitrarySampler}}
\newcommand{\certify}{\textsc{Certify}}
\newcommand{\cp}{\texttt{CP}}
\newcommand{\allclosepairs}{\texttt{AllClosePairs}}

\algnewcommand{\LineComment}[1]{\State \(\triangleright\) #1}

\newtheorem{theorem}{Theorem}
\newtheorem*{theorem*}{Theorem}
\newtheorem{lemma}{Lemma}[section]
\newtheorem{claim}[lemma]{Claim}

\newtheorem{fact}[lemma]{Fact}
\newtheorem{corollary}[lemma]{Corollary}

\newtheorem{question}{Question}

\newtheorem{definition}[lemma]{Definition}

\theoremstyle{definition}
\newtheorem{observation}[theorem]{Observation}

\theoremstyle{plain}

\newcommand{\cE}{\mathcal E}
\newcommand{\bbZ}{\mathbb Z}
\newcommand{\bbR}{\mathbb R}

\newcommand{\cL}{\mathcal L}
\newcommand{\cF}{\mathcal F}

\newcommand{\one}{\mathbbm 1}

\newcommand{\set}[1]{\left \{ #1 \right \}}

\newcommand{\cond}{\;|\;}

\renewcommand{\epsilon}{\varepsilon}

\newcommand{\ex}[1]{\mathop{{\bf E}\left[ #1 \right]}}
\newcommand{\pr}[1]{\operatorname{{\bf Pr}}\left[ #1 \right]}

\newcommand{\boldeta}{\bm{\eta}}

\newcommand{\bS}{\mathbf{S}}
\newcommand{\bT}{\mathbf{T}}

\newcommand{\bY}{\mathbf{Y}}
\newcommand{\bZ}{\mathbf{Z}}

\newcommand{\ba}{\boldsymbol{a}}
\newcommand{\bb}{\boldsymbol{b}}

\newcommand{\bg}{\boldsymbol{g}}
\newcommand{\bh}{\boldsymbol{h}}
\newcommand{\bi}{\boldsymbol{i}}
\newcommand{\bj}{\boldsymbol{j}}

\newcommand{\bu}{\boldsymbol{u}}
\newcommand{\bv}{\boldsymbol{v}}
\newcommand{\bw}{\boldsymbol{w}}
\newcommand{\bx}{\boldsymbol{x}}
\newcommand{\by}{\boldsymbol{y}}
\newcommand{\bz}{\boldsymbol{z}}

\newcommand{\balpha}{\boldsymbol{\alpha}}
\newcommand{\bbeta}{\boldsymbol{\beta}}
\newcommand{\bgamma}{\boldsymbol{\gamma}}
\newcommand{\bsigma}{\boldsymbol{\sigma}}

\newcommand{\cG}{\mathcal{G}}
\newcommand{\cH}{\mathcal{H}}
\newcommand{\cP}{\mathcal{P}}
\newcommand{\cI}{\mathcal{I}}

\usepackage{waingarten}

\title{Approximating High-Dimensional Earth Mover's Distance as Fast as Closest Pair}

\newtoggle{anonymous}
\togglefalse{anonymous}

\iftoggle{anonymous}{
\author{Anonymous Authors}
}{
\author[1]{Lorenzo Beretta}
\author[2]{Vincent Cohen-Addad}
\author[3]{Rajesh Jayaram}
\author[4]{Erik Waingarten}
\affil[1]{University of California, Santa Cruz, \texttt{lorenzo2beretta@gmail.com}}
\affil[2]{Google Research NYC \texttt{cohenaddad@google.com}}
\affil[3]{Google Research NYC \texttt{rkjayaram@google.com}}
\affil[4]{University of Pennsylvania \texttt{ewaingar@seas.upenn.edu}}
}

\usepackage{todonotes}
\date{}

\begin{document}

\maketitle
\thispagestyle{empty}

\begin{abstract}
We give a reduction from $(1+\eps)$-approximate Earth Mover's Distance (EMD) to $(1+\epsilon)$-approximate Closest Pair (CP). As a consequence, we improve the fastest known approximation algorithm for high-dimensional EMD. Here, given $p\in [1, 2]$ and two sets of $n$ points $X,Y \subset (\R^d,\ell_p)$, their EMD is the minimum cost of a perfect matching between $X$ and $Y$, where the cost of matching two vectors is their $\ell_p$ distance. Further, CP is the basic problem of finding a pair of points realizing $\min_{x \in X, y\in Y} ||x-y||_p$.
Our contribution is twofold:
\begin{itemize}
    \item 
    We show that if a $(1+\epsilon)$-approximate 
    CP can be computed in time $n^{2-\phi}$, then a $1+O(\epsilon)$ approximation to EMD can be computed in time $n^{2-\Omega(\phi)}$.  
    
    \item 
    Plugging in the fastest known algorithm for CP \cite{ACW16}, we obtain a $(1+\epsilon)$-approximation algorithm for EMD running in time $n^{2-\tilde{\Omega}(\epsilon^{1/3})}$ for high-dimensional point sets, which improves over the prior fastest running time of $n^{2-\Omega(\eps^2)}$ \cite{AZ23}.
\end{itemize} 
Our main technical contribution is a sublinear implementation of the Multiplicative Weights Update framework for EMD. Specifically, we demonstrate that the updates can be executed without ever explicitly computing or storing the weights; instead, we exploit the underlying geometric structure to perform the updates implicitly.

\end{abstract}

\newpage
\setcounter{page}{1}

\newcommand{\citeplaceholder}[1]{[#1]} 

\section{Introduction}

In this paper, we study algorithms for the Earth Mover's Distance ($\EMD$)
which is a fundamental measure of dissimilarity between two sets of point in a metric space \cite{villani2008optimal, santambrogio2015optimal, PC19} (also known as the Wasserstein-1 metric). Specifically, we consider the ground metric $(\R^d, \ell_p)$, for $p \in [1,2]$, 
and given two subsets of $n$ vectors, $X = \{ x_1,\dots, x_n \}, Y = \{ y_1,\dots, y_n\} \subset \R^d$, the Earth Mover's distance ($\EMD$) between $X$ and $Y$ is
\[ \EMD(X, Y) = \min_{\substack{\pi \colon [n]\to[n] \\ \text{bijection}}} \hspace{0.2cm} \sum_{i=1}^n \|x_i - y_{\pi(i)} \|_{p} \]
In other words, $\EMD(x, y)$ is the value of a min-cost flow in the complete bipartite geometric graph: the $n$ points of $X$ become vertices on the left-hand side and have uniform supply, the $n$ points of $Y$ become vertices on the right-hand side and have uniform demand, and every edge between $x_i$ and $y_j$ appears with cost $\|x_i - y_j\|_p$ and infinite capacity.

A natural algorithmic approach to compute $\EMD$ is to first construct the geometric graph and then utilize the best graph algorithm for min-cost flow. However, this graph has $n^2$ edges, resulting in an at least quadratic time algorithm. On the other hand, the input has size $O(nd)$, thus it is feasible that faster algorithms exist. In fact, there is a extensive line-of-work on approximation algorithms for $\EMD$ running in sub-quadratic time; to avoid constructing the complete graph, algorithms must exploit the geometric structure of $(\R^d,\ell_p)$~\cite{C02, IT03, AIK08,SA12,agarwal2014approximation, ANOY14, AFPVX17,KNP19,rohatgi2019conditional, ACRX22,AZ23, BR24, CCRW23, FL23}. 

\paragraph{The Spanner Approach.} A standard approach to avoid constructing the full geometric graph is to employ \textit{geometric spanners}, which are sparse graphs where the shortest path distances approximates the geometric distances. Since the runtime of min-cost flow graph algorithms scale with the number of edges~\cite{S17,ASZ20, L20,CKLPPS22}, sparser spanners will yield faster approximations for $\EMD$.\footnote{Using the almost-linear time algorithm for capacitated min-cost flow~\cite{CKLPPS22}, it suffices to construct a \emph{directed} spanner (i.e., a directed graph whose shortest paths approximate pairwise distances). If using near-linear, approximate, uncapacitated min-cost flow~\cite{S17, ASZ20, CKLPPS22, L20}, the spanner must be undirected.} 
In low-dimensional space, where $d = o(\log n)$, one can construct $(1+\eps$) approximate spanners with $n \cdot O(\eps^{-d})$ edges \cite{clarkson1987approximation, althofer1993sparse,ruppert1991approximating}. 
  However,  until recently it was not known whether sub-quadratic $(1+\eps)$ approximate spanners exist for high-dimensional spaces for $\eps < 2.4$.\footnote{For large constant approximations, it is known that $C$-approximate $n^{1+O(1/C)}$ sized spanners exist \cite{HIS13}. } This was addressed by~\cite{AZ23}, who designed such spanners (crucially using additional Steiner vertices) with $n^{2-\Omega(\eps^3)}$ undirected edges or $n^{2-\Omega(\eps^2)}$ directed edges. The main application of their work is a $(1+\eps)$-approximation algorithm for $\EMD$ over $(\R^d,\ell_p)$ running in time $n^{2 - \Omega(\eps^2)}$. To date, this is the only sub-quadratic $(1+\eps)$ approximation of high-dimensional $\EMD$.

  The fact that the runtimes of all state-of-the-art $\EMD$ approximations are directly tied to the best available spanner construction motivates our central question: can this dependency be circumvented? Specifically:

    \begin{quote} \centering 
    {\it
Do there exists $(1+\eps)$-approximation algorithms to high-dimensional $\EMD$ with running times faster than that of the best spanner constructions?}
\end{quote}

We answer the above question affirmatively, and thereby improve the runtime of best known $(1+\eps)$ approximation for $\EMD$. Our key technical contribution enabling this is a black-box reduction from $\EMD$ to the approximate closest pair (CP) problem. Here, the CP problems is: given two sets $X,Y \subseteq \R^d$, return a pair $(x,y) \in X \times Y$ such that $\|x-y\|_1 \leq (1+\eps) \min_{x' \in X, y' \in Y} \|x'-y'\|_1$.\footnote{For both $\EMD$ and CP, we focus on the $\ell_1$ norm. Via standard reductions \cite{johnson1982embedding}, one can embed $\ell_p$ for $p \in [1,2]$ into $\ell_1$ with $(1+\eps)$ error, making $\ell_1$ the more general space.} Closest pair is a cornerstone problem in computational geometry \cite{,,V15,ACW16}, and is intimately related to the nearest neighbor search   problem \cite{IM98,HIM12,AI08,A09}.
Crucially, CP is amenable to powerful algorithmic techniques outside of spanners, most notably the polynomial method~\cite{ACW16}. These techniques have yielded CP algorithms with runtimes faster than the best known spanner constructions.  Specifically, we prove the following:

\begin{theorem}[Main Theorem, informal version of \Cref{thm:main theorem formal}]\label{thm:main} 
If there exists an algorithm for $(1+\varepsilon)$-approximate closest pair on $(\mathbb{R}^d, \ell_1)$ with running time $T_{CP}(n, \varepsilon) = n^{2-\phi(\varepsilon)}$, then there exists an algorithm that computes a $(1+O(\varepsilon))$-approximation to $\EMD(A, B)$ over $(\mathbb{R}^d, \ell_p)$ for any $p \in [1,2]$ in time $n^{2-\Omega(\phi(\varepsilon))}$.
\end{theorem}

Theorem \ref{thm:main} provides a fine-grained reduction, translating runtime improvements for CP directly into runtime improvements for $\EMD$, losing only a constant in $\eps$ and $\phi(\eps)$. The existence of such an efficient reduction is perhaps surprising. Namely, while exact closest pair admits a simple $O(n^2)$ algorithm, computing $\EMD$, which is a global matching optimization task, appears significantly more complex. Indeed, a natural greedy approach for $\EMD$---repeatedly finding the closest available pair $(a,b)$, matching them, and removing them---is known to yield very poor approximations~\cite{reingold1981greedy}. Despite the failure of this simple reduction,
Theorem \ref{thm:main} demonstrates that, via a considerably more involved reduction, approximate $\EMD$ is essentially no harder than approximate CP.

By plugging in the fastest known algorithm for approximate CP in high-dimensional $\ell_1$ by Alman, Chan, and Williams \cite{ACW16}, which runs in time $n^{2 - \tilde{\Omega}(\eps^{1/3})}$, we obtain:

\begin{corollary}\label{cor:main}
There exists an algorithm that, with high probability, computes a $(1+\varepsilon)$-approximation to $\EMD(A, B)$ between two $n$-point sets $A, B \subset (\mathbb{R}^d, \ell_1)$ in time $dn + n^{2-\Omega(\varepsilon^{1/3}/\log(1/\varepsilon))}$.
\end{corollary}
This is a significant improvement over the $n^{2 - O(\eps^2)}$-time spanner-based algorithm of \cite{AZ23}, reducing the exponent of $\epsilon$ by a factor of $6$. Moreover, we note that it is perhaps unlikely that the complexity of Corollary \ref{cor:main} can be matched by the Spanners approach. Specifically, there is a non-trivial barrier to improving the sparsity of $(1+\eps)$ high-dimensional spanners beyond $n^{2-\eps}$; specifically, all such spanner constructions are crucially based on Locality Sensitive Hashing (LSH)~\cite{I01a,HIS13,AZ23}, and therefore only as efficient as the best algorithm for all-pairs nearest neighbor search via LSH, for which the best known constructions require $n^{2-O(\epsilon)}$ time. Thus, improving the spanners approach beyond $n^{2-O(\epsilon)}$ would likely require a major breakthrough in LSH algorithms. 

\subsection{High-Level Overview and Conceptual Contribution} \label{sec:high-level approach}
We now provide a high-level overview of our algorithm, and elaborate on additional details in the Technical Overview (Section \ref{sec:tech-overview}). 
Our approach to Theorem \ref{thm:main} is to solve the standard transshipment linear program (Figure \ref{fig:LP}). 
However, a direct approach faces an immediate obstacle: the primal linear program  has $O(n^2)$ variables, making it impossible to even represent a feasible solution within our target subquadratic runtime. While the dual program offers a potential advantage, as it has only $O(n)$ variables (allowing representation of a dual solution), it presents its own challenge: checking dual feasibility requires evaluating $O(n^2)$ constraints, which again is too expensive. Despite this strong restriction, we nonetheless demonstrate how we can implement the multiplicative weights update (MWU) framework to solve the dual in time sublinear in the number of constraints.

We emphasize that while many works have used the MWU framework to solve linear programs, including for EMD \cite{L20, Z23}, these works typically assume the ability to explicitly store and manipulate the MWU weights. To our knowledge, our work is the first to tackle the setting where, due to the strict subquadratic time requirement, the weights cannot be stored explicitly. Instead, our algorithm must continuously operate on the implicit geometric representation of the input points. This necessitates a careful integration of geometric techniques and optimization principles.

Specifically, in Section \ref{sec:template} we show that the MWU update step can be reduced to the task of generating $\tilde{O}(n)$ samples from a class of distributions $\lambda$ defined over over pairs $[n] \times [n]$. The probability $\lambda_{i,j}$ is proportional to $\exp\left(K_{i,j} \cdot \frac{1}{\|x_i - y_j\|_1}\right)$, where $K_{i,j} > 0$ is a value depends on the current dual variables. Again, we emphasize that we cannot \emph{explicitly} represent this distribution due to its size. Instead, we must rely solely 
 on the implicit geometric representation of the points. Our key insights for sampling efficiently from $\lambda$ are as follows:

\begin{enumerate}
    \item \textbf{Reducing to fixed $K$:} We first demonstrate that we can partition $X \times Y$ into combinatorial rectangles $\{X_1 \times Y_1,\dots,X_r \times Y_r\}$ so that $K_{i,j}$ is (roughly) the same for all $(i,j)$ in a rectangle. By sampling from each rectangle $X_i \times Y_j$ independently, this allows us to reduce to the case where $K_{i,j} = K$ is fixed for all pairs $i,j$ (Lemma \ref{lem:from constant D to arbitrary D}).

    \item \textbf{Ideal Sampling with Distance Gaps:} Hypothetically, suppose that every distance was of the form $(1+\eps)^t$ for an integer $t$, creating discrete gaps. Since $K$ is fixed (from step 1), $\lambda_{i,j} \propto \exp\left( K / \|x_i-y_j\|_1\right)$, so 
    we could use an $(1+\eps)$-approximate CP algorithm to iteratively find the pairs with smallest distances (thus highest probabilities $\lambda_{i,j}$). The CP algorithm can stop after it has reached a level $t$ such that there are at least $n^{1+\phi}$ pairs at distance $(1+\eps)^{t+1}$. After explicitly storing the pairs with distance at most $(1+\eps)^t$, we can use rejection sampling for the remaining pairs. Rejection sampling will be efficient ($\tilde{O}(n^{1-\phi})$ time per sample) because the distance gaps guarantee that a $1/n^{1-\phi}$ fraction of remaining points all share the \emph{next} highest probability level.
    
    \item \textbf{Challenge: Real Distances $\&$ Approximate CP:} Note that simply rounding the distances in MWU to the nearest power of $(1+\eps)$ fails. The core issue is that the approximate CP algorithm works on the original geometry of the points, not these conceptual rounded distances.
 Namely, a $(1+\eps)$ approximate CP algorithm may return pairs with (rounded) distance $(1+\eps)^{t+1}$ before it finds all pairs with (rounded) distance $(1+\eps)^t$. This is a critical issue since $\lambda_{i,j}$ depends exponentially on the distance (and $K$ can be $\polylog(n)$); thus, missing even one closer pair can unacceptably skew the sampling distribution.
 
    \item \textbf{Flawed Attempt: CP-Adaptive Rounding:} 
    To address this issue,
    we could try rounding distances adaptively based on the CP output within one iteration (e.g., round down if the pair is found by CP, round up otherwise).
    Thus, undiscovered points have lower probabilities by construction. This approach has a fundamental flaw: the rectangles $X_k \times  Y_k$ which are input to the CP algorithm will change after each MWU iteration; in particular, in one iteration the CP algorithm may find $(i,j)$ and round it down, and on the next iteration the CP algorithm may not find the pair. 
    Thus, we must decide on a fixed rounding scheme up front.
    
     \item \textbf{Solution: Fixed Randomized Rounding:} 
     Our key technique is showing the success of a randomized rounding strategy. Specifically, we sample a set  $S \subset X \times Y$ of size $n^{2-\phi/2}$; we then round distances down for pairs $(i,j) \in S$, and round distances up otherwise. 

     \begin{itemize}
         \item \textbf{Why it works:} For any fixed set $L$ of $n^{1+\phi}$ pairs, the intersection $|S \cap L|$ will be sufficiently large (at least $n^{1+\phi/2}$) with high probability. If $L$ is the set of pairs with distance between $(1+\eps)^t$ and $(1+\eps)^{t+1}$, then by rounding down the distances in $S \cap L$ we ensure there's a sufficiently large pool of points with the ``maximal'' probability among those not explicitly found by CP, making rejection sampling efficient (needed in step 2). 
         \item \textbf{Addressing Critical Dependencies:}  This seems promising, except for a potentially fatal fault:  $S$ influences $\lambda$, which influences the samples from $\lambda$, which influence the dual variables and thus the rectangles $X_k \times Y_k$ that are constructed in the next iteration. Thus, the samples $S$ are not independent of the sets $L$ (as defined by the pairs at a given distance inside of a rectangle) we need them to intersect with. However, the crucial observation is that the rectangles depend \textit{only} on the $\tilde{O}(n)$ samples from $\lambda$ on the previous step, not all $n^2$ probabilities affected by $S$. Since there are only $T = \polylog(n)$ iterations, the number of possible sample histories (and thus possible rectangle partitions) is at most $2^{T \cdot \tilde{O}(n)}$. We can thus union bound over all these possibilities to show that our single, fixed random sample $S$ works with high probability for any set $L$ arising within any rectangle that the algorithm actually encounters (details in Section \ref{sec:consistent rounding}).
         
     \end{itemize}

\end{enumerate}

\section{Technical Overview} \label{sec:tech-overview}
Earlier in Section \ref{sec:high-level approach} we gave a bird-eye view of our algorithmic approach. In this section, we sketch the techniques used to prove \Cref{thm:main} in further detail. We first devise a template algorithm that, if implemented naively, runs in quadratic time. Our main technical contribution will be to show how this template algorithm can be implemented via a small number calls to a closest pair algorithm. We will then exploit the existence of fast algorithms for approximate closest pair \cite{ACW16}.

Fix two input sets
 $X = \{x_1 \dots x_n\}, Y = \{y_1 \dots y_n\} \subseteq ([1, \Phi]^d, \ell_1)$. We extend the definition of $\EMD$ 
 to arbitrary supply vectors and denote with $\emd(\mu, \nu)$ the cost of minimum uncapacited flow routing the supply $\mu_1 \dots \mu_n \in [0,1]$ located at $x_1 \dots x_n$ to demands $\nu_1 \dots \nu_n\in [0,1]$ located at $y_1 \dots y_n$. In particular, the $\EMD$ between $X$ and $Y$ can be denoted by $\emd(\ind, \ind)$, or simply $\emd$.

 In Section \ref{sec:templateintro}, we will describe the template algorithm for approximating $\EMD$, based on solving the dual LP via multiplicative weights update (MWU). The key step to implementing the update step in the MWU procedure is to sample from the distributions $\lambda$ described inSection \ref{sec:high-level approach}. 
 Then, in Section \ref{sec:subquadratic-intro}, we describe our scheme to sample from this class of distributions in subquadratic time.

\subsection{Template Algorithm for Approximating $\EMD$}\label{sec:templateintro}

Our starting point for approximating $\EMD$ will be to solve the standard linear program for the bipartite min-cost flow problem, shown below in \Cref{fig:LP}.

\begin{figure}[H]
\begin{framed}
\begin{minipage}[t]{0.1\textwidth}
\vspace{1mm}
\textbf{Primal}
\begin{flalign*}
&\text{Minimize} &\sum_{i, j \in [n]} \bgamma_{i, j} \cdot ||x_i - y_j|| \\
&\text{subject to} &\sum_{j\in [n]} \bgamma_{i, j} = 1 \quad \forall i \in [n] \\
& &\sum_{i \in [n]} \bgamma_{i, j} = 1 \quad \forall j \in [n] \\
& & \bgamma_{i, j} \geq 0 \quad \forall i,j \in [n]
\end{flalign*}
\end{minipage}
\hfill\vline\hfill
\begin{minipage}[t]{0.45\textwidth}
\vspace{1mm}
\textbf{Dual}
\begin{align*}
&\text{Maximize} &\sum_{i \in [n]} \balpha_i - \bbeta_i &\\
&\text{subject to} & \frac{\balpha_i - \bbeta_j}{||x_i - y_j||} \leq 1 \quad & \forall i, j \in [n]
\end{align*}
\end{minipage}
\end{framed}
\caption{Linear program for $\EMD$.}
\label{fig:LP}
\end{figure}

Our goal will be to approximate the optimal objective value to this LP in subquadratic time. However, notice that the primal formulation has $n^2$ variables, in the form of the flows $\{\bgamma_{ij}\}_{i, j \in [n]}$. Thus, we cannot even write down a solution to the primal, let alone directly implement iterative update methods. On the other hand, the dual has only $2n$ variables $\{\balpha_i, \bbeta_i\}_{i \in [n]}$.

Thus, in principle it could possible to compute a dual solution $(\balpha, \bbeta)$ in subquadratic time, since its size is only $O(n)$. Of course, the challenge is now that the dual has $n^2$ constraints, thus \emph{verifying} whether a given solution $(\balpha, \bbeta)$ is feasible is likely impossible in subquadratic time. However, in what follows we will show that we can approximately verify feasibility efficiently.

For now, we first we describe an ``inefficient`` algorithm that computes a $(1+\varepsilon)$-approximation of 
$\emd$
in time $\tilde O(n^2)$.

Let $C_{ij} := ||x_i - y_j||$. For $t > 0$, consider the following polytope:
\[ 
\Gamma_t = \left\{ (\balpha, \bbeta) \in \R^{n} \times \R^n  : \frac{1}{t} \cdot \left( \sum_{i=1}^n \balpha_i - \sum_{j=1}^n \bbeta_j \right) - \max_{i,j} \frac{|\balpha_i - \bbeta_j|}{C_{ij}} > 0 \right\}. 
\]
By duality, it is straightforward to prove that 
$\emd > t$
if and only if $\Gamma_t \neq \emptyset$.
By performing an search over $t$ in powers of $(1+\eps)$, the problem of approximating 
$\emd$
can be solved by leveraging a procedure that, if $\Gamma_{(1 + \varepsilon)  t} \neq \emptyset$, returns a certificate $(\balpha, \bbeta) \in \Gamma_t$. 
To this end, we employ the framework of multiplicative weights update (MWU), described in the following. 

\paragraph{Multiplicative Weights Update (MWU).}
Observe that the polytope $\Gamma_t$ is defined by $2n^2$ constraints (each absolute value constraint corresponds to two linear constraints), which we index with triplets $(i, j, \sigma) \in [n] \times [n] \times \{\pm 1\}$.
We define a set of positive weights $w_{i, j, \sigma}$, one for each constraints, and initialize all of them to $1$.
This induces a distribution over constraints $\lambda$ such that $\lambda_{i, j, \sigma}$ is proportional to $w_{i, j, \sigma}$. Importantly, again, we note that we cannot explicitly maintain the distribution $\lambda$, as it is a $O(n^2)$ sized object. Nevertheless, in what follows we show that it will suffice to sample from $\lambda$ to compute the MWU update.

The classic MWU procedure operates in rounds: at each round, we define a ``special`` constraint (\ref{eq:special constraint}) by taking a $\lambda$-weighted combination of the constraints 
\begin{equation} \label{eq:special constraint}
\frac 1 t \cdot \left( \sum_{i=1}^n \balpha_i - \sum_{j=1}^n \bbeta_j\right) - \sum_{i=1}^n \sum_{j=1}^n \sum_{\sigma \in \{-1,1\}} \lambda_{i,j,\sigma} \frac{\sigma(\balpha_i - \bbeta_j)}{C_{ij}} > 0,
\end{equation}
then, we design an oracle $\textsc{Certify}(\lambda)$ that should return $(\balpha, \bbeta)$ such that: 
\begin{enumerate}[\textbf{Invariant} (I):]
    \item $(\balpha, \bbeta)$ satisfies the special constraint in \Cref{eq:special constraint}. \label{item:special constraint invariant}
    \item $\bigg |\frac 1 t \left( \sum_{\ell=1}^n \balpha_\ell - \sum_{\ell=1}^n \bbeta_\ell \right)  \bigg|, \bigg|  \frac{\balpha_i - \bbeta_j}{C_{ij}} \bigg| \leq \polylog(n)$ for all $i, j \in [n]$.
    \label{item:invariant linfinity norm}
\end{enumerate}
Then, the weights are updated to penalize constraints that are violated by $(\balpha, \bbeta)$. 
 
\begin{equation} \label{eq:mwu update rule}
w_{i, j, \sigma} \leftarrow w_{i, j, \sigma} \cdot \exp\left(\eta \cdot \left( \frac 1 t \cdot \left( \sum_{i=1}^n \balpha_i - \sum_{j=1}^n \bbeta_j\right) -  \frac{\sigma\cdot (\balpha_i - \bbeta_j)}{C_{ij}} \right) \right)
\quad \text{ for constant } \eta > 0.
\end{equation}

After $\sfR = \polylog(n)$ rounds, we show that if the two invariants are always satisfied by the output of the $\textsc{Certify}(\lambda)$ procedure, then the average $(\bar \balpha, \bar \bbeta)$ of the dual variables $(\balpha^{(1)}, \bbeta^{(1)})\allowdisplaybreaks, \dots , \allowbreak (\balpha^{(\sfR)}, \bbeta^{(\sfR)})$ returned at each round satisfies the constraints in $\Gamma_{(1- 2\varepsilon) t}$ (Lemma \ref{lem:mwu}).
We are left to implement the oracle $\textsc{Certify}(\lambda)$.

\paragraph{Probabilistic tree embedding.}
Before we implement $\textsc{Certify}(\lambda)$, we first recall a fundamental tool in high-dimensional geometry: the probabilistic tree embedding. Specifically, this technology allows us to embed points in $(\R^d, \ell_1)$ into a tree metric with polylogarithmic \emph{expected} distortion \cite{AIK08}.
Concretely, given the set of points $X \cup Y \subseteq ([1, \Phi]^d, \ell_1)$, there exists a distribution over trees $\calT$ satisfying 1--5 below: 
\begin{enumerate}
    \item Each leaf of $\calT$ corresponds to an element of $X \cup Y$.
    \item $\calT$ has depth $\log \Phi$.
    \item Every edge connecting a node at level $\ell$ with one at level $\ell + 1$ has weight $\frac \Phi {2^\ell}$.
    \item $||x_i - y_j|| \leq \Theta(\log n) \cdot \sfd_\calT(x_i, y_j)$, with  probability $1-1/\poly(n)$ (where $\sfd_\calT$ is the distance in the tree). \label{item:whp inequality}
    \item $\Ex[\sfd_\calT(x_i, y_j)] \leq \Theta(\log \Phi) \cdot ||x_i - y_j||$. \label{item:expected inequality}
\end{enumerate}
\begin{framed}\noindent
We remark that probabilistic tree embeddings do \emph{not} satisfy Bi-Lipschitz distortion guarantees. Namely, we \emph{cannot} guarantee that with high probability:
\begin{equation}\label{item:better guarantee}
    \sfd_\calT(x_i, y_j) \leq \Theta(\log \Phi) \cdot ||x_i - y_j||
\end{equation}
This deficit will become relevant shortly.
\end{framed}

For the purposes of approximating $\EMD$ to a $\polylog(n,\Phi)$ factor, however, it turns out that expected $\polylog$ distortion suffices.
Specifically, it is not difficult to show that the $\EMD$ in the tree metric, $\emd_\calT(\mu, \nu)$, enjoys the same distortion guarantees: $\emd_\calT(\mu, \nu) \geq \Theta(1/\log n) \cdot \emd(\mu, \nu)$; $\Ex[\emd_\calT(\mu, \nu)] \leq \Theta(\log \Phi) \cdot \emd(\mu, \nu)$.
The advantage of this embedding is that now the $\EMD$ within the tree metric is relatively simple to compute, as it can be written exactly as:

\begin{equation} \label{eq:emd on tree embeddingss}
\emd_\calT(\mu, \nu) = \sum_{v \in \calT}  \frac{\Phi}{2^\ell} \cdot \bigg|\sum_{x_i \in v} \mu_i -  \sum_{y_i \in v} \nu_i \bigg|,
\end{equation}
where $x \in v$ means that $x$ belong to the subtree rooted at the internal node $v$. The formula (\ref{eq:emd on tree embeddingss}) will be useful in our implementation of $\textsc{Certify}(\lambda)$ below. Specifically, we remark that while alternative methods may exist for designing $\polylog$ ``rough'' approximations to $\EMD$, the tree-embedding formulation, specifically (\ref{eq:emd on tree embeddingss}), will be critical in allowing $\emd_\calT(\mu, \nu)$ to be estimated via sampling.

\paragraph{$\textsc{Certify}(\lambda)$ implementation.}
We now describe an \emph{inefficient} implementation of the $\textsc{Certify}(\lambda)$ procedure, which will run in quadratic time; we will show thereafter how this procedure can be approximated by a subquadratic sampling-based algorithm.

To construct  feasible variables $(\balpha, \bbeta)$, we first define the following supply and demand vectors:
\[
\mu_i = t \sum_{j=1}^n \sum_{\sigma \in \{-1,1\}} \dfrac{\lambda_{i,j,\sigma} \cdot \sigma}{C_{ij}} \qquad \text{and}\qquad \nu_j = t \sum_{i=1}^n \sum_{\sigma \in \{-1,1\}} \dfrac{\lambda_{i,j,\sigma} \cdot \sigma}{C_{ij}}.
\]

Observe that $\gamma_{ij} = t (\lambda_{i,j,1} - \lambda_{i,j,-1}) / C_{ij}$ is a feasible flow
for $\EMD(\mu, \nu)$ 
with cost at most $t$. Thus, $\EMD(\mu, \nu) \leq t$.
Moreover, if we are promised that $\Gamma_{(1+\varepsilon) t} \neq \emptyset$, 
then $\emd(\ind, \ind) > (1+\varepsilon) t$. Hence, using a ``triangle inequality'' for $\EMD$:
\[
\Theta(\log n) \cdot \emd_{\calT}(\ind - \mu, \ind - \nu)  \geq \emd(\ind - \mu, \ind - \nu)  \geq \emd(\ind, \ind) - \emd(\mu, \nu) > \varepsilon \cdot t.
\]
By \Cref{eq:emd on tree embeddingss} applied to $(\ind - \mu, \ind - \nu)$, there exists a level $T_\ell \subseteq \calT$ of the tree $\calT$ such that
\begin{align} \label{eq:heavy level}
 \sum_{v \in T_{\ell}} \frac{\Phi}{2^{\ell}} \left| |X_v| - |Y_v| - t \sum_{\substack{i \in [n] \\ x_i \in v}} \sum_{j=1}^n \sum_{\sigma \in \{-1,1\}} \dfrac{\lambda_{i,j,\sigma} \cdot \sigma}{C_{ij}} + t \sum_{\substack{j \in [n] \\ y_j \in v}} \sum_{i=1}^n \sum_{\sigma \in \{-1,1\}} \frac{\lambda_{i,j,\sigma} \cdot \sigma}{C_{ij}} \right| = \tilde \Omega(\varepsilon t).
 \end{align}
 Set $\balpha^*_i = \bbeta^*_j = \frac{\Phi}{2^\ell} \cdot s_v$ for $x_i, y_j \in v$, where
 $s_v \in \{-1,1\}$ denote the sign which realizes the absolute value on the left-hand side of \Cref{eq:heavy level}.
 With some additional work, we can rearrange the terms in \Cref{eq:heavy level} to obtain
 \[
\sum_{i=1}^n \balpha_i^* - \sum_{j=1}^n \bbeta_j^* - t \sum_{i=1}^n \sum_{j=1}^n \sum_{\sigma \in \{-1,1\}} \lambda_{i,j,\sigma} \cdot \sigma \cdot \dfrac{\balpha_i^* - \bbeta_j^*}{C_{ij}} = \tilde \Omega(\varepsilon t),
 \]
so $(\balpha^*, \bbeta^*)$ satisfies invariant (\ref{item:special constraint invariant}).

\paragraph{Challenge: satisfying invariant (\ref{item:invariant linfinity norm}) requires Bi-Lipschitz tree embeddings.}
Notice that by construction $\balpha^*_i - \balpha^*_j \neq 0$ only if $x_i$ and $y_j$ belong to different subtrees at level $\ell$.
If Property \ref{item:better guarantee} held for the tree $\calT$, then $\balpha^*_i - \balpha^*_j \neq 0$ would imply $C_{ij} = \tilde \Omega(\frac{\Phi}{2^\ell})$, and, in turn, $(\balpha^*, \bbeta^*)$ would satisfy invariant (\ref{item:invariant linfinity norm}).
Unfortunately, there exists a set of points that cannot be embedded into a tree satisfying both properties \ref{item:whp inequality} and \ref{item:better guarantee} for all pairs\footnote{For example the set of $n$ evenly spaced points on a circumference.}. Thus, it will not be possible to achieve invariant (\ref{item:invariant linfinity norm}) without a \textit{Bi-Lipschitz} guarantee on the distortion of points. Moreover, satisfying invariant (\ref{item:invariant linfinity norm}) is critical to the success of the MWU procedure.

\paragraph{Obtaining Bi-Lipschitz guarantees via post-hoc perturbation.}
We address this challenge by designing a post-hoc perturbation $X' \cup Y'$ of the points in $X\cup Y$ such that:
\begin{itemize}
    \item On $X' \cup Y'$, $\calT$ satisfies properties 1--5, as well as the stronger property in Equation (\ref{item:better guarantee}).
    \item For any fixed $\mu, \nu$, $\emd_{X', Y'}(\mu, \nu) = (1\pm \varepsilon)\cdot \emd_{X, Y}(\mu, \nu)$.
\end{itemize}

Importantly, the perturbed points $X', Y'$ \emph{depends} on the realization of the random tree $\calT$.
Specifically, for each node $\bv \in \calT$ at level $\ell$ we draw a random $\bz_{\bv} \in [0, \varepsilon \cdot \Phi / 2^\ell ]^{O(\log n)}$.
For each point $\bx \in X \cup Y$ we compute the root-to-leaf path $\bv_1 \dots \bv_h = \bx$ and add the vector $\sum_{i=1}^h \bz_i$ to $\bx$.
Any two points whose least common ancestor is at level $\ell$ will have their distance increased by $\approx \frac{\Phi}{2^\ell}$.
At the same time, the perturbation cannot make $\emd_{X, Y}(\mu, \nu)$ much larger than $\emd_{\calT}(\mu, \nu)$, which, in turn, is at most $\Theta(\log \Phi) \cdot \emd_{X, Y}(\mu, \nu)$.

\subsection{Implementing the Template in Subquadratic Time}\label{sec:subquadratic-intro}
We now demonstrate how the above algorithmic template can be efficiently implemented. In particular, we first demonstrate how the the oracle $\textsc{Certify}(\lambda)$ can be implemented using $\tilde{O}(n)$ samples from $\lambda$, rather than requiring an explicit representation of all $2n^2$ probabilities. Secondly, we reduce the problem of sampling from $\lambda$ to the problem of finding an approximate closest pair.

\paragraph{Implementing $\certify(\lambda)$ with samples.}
First, we recall how the previously described $\textsc{Certify}(\lambda)$ procedure constructs $(\balpha^*, \bbeta^*)$. Let
\begin{gather} \label{eq:estimate the sign}
\calQ_v(\lambda) := \frac{\Phi}{2^\ell} \cdot \left(|X_v| - |Y_v| - t \sum_{\substack{i \in [n] \\ x_i \in v}} \sum_{j=1}^n \sum_{\sigma \in \{-1,1\}} \dfrac{\lambda_{i,j,\sigma} \cdot \sigma}{C_{ij}} + t \sum_{\substack{j \in [n] \\ y_j \in v}} \sum_{i=1}^n \sum_{\sigma \in \{-1,1\}} \frac{\lambda_{i,j,\sigma} \cdot \sigma}{C_{ij}} \right)
\end{gather}
First, we find a level $T_\ell \subseteq \calT$ in the probabilistic tree satisfying \Cref{eq:heavy level}, which can be rewritten as $\sum_{v \in T_\ell} |\calQ_v(\lambda)| = \tilde \Omega(\varepsilon t)$.
Then, define $\balpha^*_i = \bbeta^*_j = \frac{\Phi}{2^\ell} \cdot \text{sign}\left( \calQ_v(\lambda) \right)$ for all $x_i, y_j \in v \in T_\ell$.

To implement $\textsc{Certify}(\lambda)$ with samples, we draw $s = \tilde O(n)$ i.i.d. samples $(i, j, \sigma) \sim \lambda$, define $\hat \lambda$ as the empirical distribution of these samples, and run $\textsc{Certify}(\hat \lambda)$.
We can verify that, if we have concentration of the form $\calQ_v(\hat \lambda) \approx \calQ_v(\lambda)$, then $(\balpha^*, \bbeta^*)$ satisfies invariants \ref{item:special constraint invariant} and \ref{item:invariant linfinity norm}. Thus, the correctness of $\certify(\hat\lambda)$ boils down to bounding the variance of $\calQ_v(\hat \lambda)$ for one sample 
$(\bi_S, \bj_S, \bsigma_S) \sim \lambda$, which we do in the following.

The stochastic part of $\calQ_v(\hat \lambda)$ can be written as
\begin{equation} \label{eq:stochastic term}
t \cdot \sum_{\substack{i,j \in [n] \\ \sigma \in \{-1,1\}}} \sigma \cdot \left(\ind \{ x_{i} \in v \} - \ind\{ y_{j} \in v \} \right) \cdot \ind\{(\bi_S, \bj_S, \bsigma_S) = (i,j,\sigma)\} \cdot\frac{\Phi / 2^\ell}{C_{ij}}
\end{equation}
If $x_i, y_j \in v$ (or $x_i \notin v, y_j \notin v$), the terms in the sum
corresponding to $i,j$  in \Cref{eq:stochastic term} cancel.
Else, if $x_i \in u$ and $y_j \in v$ for two distinct nodes $u, v \in T_\ell$, then by property \ref{item:better guarantee} of $\calT$, we have
\[
\frac{\Phi / 2^\ell}{C_{ij}}  =  \frac{\Theta\left( \sfd_\calT(x_i, y_j)\right)}{C_{ij}} = O(\log \Phi).
\]
Each summand in \Cref{eq:stochastic term} has magnitude at most $O(t \cdot \log \Phi)$. Thus, $s = \tilde O(n / \eps^2)$ samples suffice to shrink $\Varx{\calQ_v}$ to $\frac{\eps^2 t^2}{n} \cdot  \Prx_{(\bf i,\bf j) \sim \lambda}[{\bf i} \in v \text{ or } {\bf j} \in v]$, which suffice for our analysis.

\paragraph{Interlude: sampling seems as hard as closest pair.}
In this interlude we suggest why closest pair seems the right tool for sampling from $\lambda$.
Let $(\balpha^{(1)}, \bbeta^{(1)}) \dots (\balpha^{(r)}, \bbeta^{(r)})$ be the dual solutions returned on the first $r$ rounds of MWU. Let $\lambda$ be the distribution proportional to the weights $w_{i, j, \sigma}$ updated $r$ times as in \Cref{eq:mwu update rule}. Then,
\begin{equation} \label{eq:lambda definition}
\lambda_{i, j, \sigma} \propto \exp\left(\frac{\sigma\cdot \left( \sum_{\ell = 1}^r \balpha^{(\ell)}_i - \bbeta^{(\ell)}_j \right)}{C_{ij}} \right).
\end{equation}
According to invariant (\ref{item:invariant linfinity norm}), we could have $|\balpha^{(\ell)}_i -\bbeta^{(\ell)}_j| \approx C_{ij} \cdot \polylog(n)$.
Suppose that $C_{ij}=1$ for all $i, j \in [n]$ besides $C_{i^*, j^*} = 1 - \varepsilon$ and that, for all $i, j \in [n]$, $\sum_{\ell = 1}^r \balpha^{(\ell)}_i - \bbeta^{(\ell)}_j = K = \polylog (n)$. Then, $\lambda_{i, j , \sigma} \propto \exp(\sigma \cdot K / C_{ij})$ is overwhelmingly concentrated on $(i^*, j^*, +1)$.
Thus, under the invariants proven so far, sampling from $\lambda$ could be as hard as finding the closest pair.
In the following, we prove that in fact sampling from $\lambda$ can be reduced to closest pair.

\subsection{Sampling from $\lambda$ via Closest Pair}
 
\paragraph{Roadmap.}
Sampling from $\lambda$ as defined in \Cref{eq:lambda definition} presents multiple challenges, which we decouple as follows.
To begin, suppose we instead were tasked with solving a restricted version of the problem:
\begin{enumerate}[\underline{Restriction} (1)]
    \item We assume that the term $\sum_{\ell = 1}^r \balpha^{(\ell)}_i - \bbeta^{(\ell)}_j$ in \Cref{eq:lambda definition} does not depend on $(i, j)$. Namely, we sample from $\lambda_{i,j, \sigma} \propto \exp(\sigma \cdot K / C_{ij})$.
    \label{item:restriction coefficient}
    \item We assume that our closest pair algorithm returns all ``close'' pairs, where a pair is close if its distance is one of the smallest $O(n^{1+\phi})$ distances.
    \label{item:restriction all close pairs}
\end{enumerate}

We will show that, with some effort,  both restrictions can ultimately be lifted.

\paragraph{Sampling from an easier distribution.}
First, note that to sample from $\lambda_{i,j, \sigma} \propto \exp(\sigma \cdot K / C_{ij})$, it is sufficient to sample a variable $\xi_{ij} \propto \exp(|K| / C_{ij})$ and thereafter apply a rejection sampling step to sample the sign $\sigma \in \{-1,1\}$, which gives a constant runtime overhead.
Interestingly, for very large $K$, $\xi$ concentrates on the closest pair, thus sampling from $\xi$ for arbitrary $K$ is as hard as \emph{exact}\footnote{
Even if we assume a realistic bound on $K$, e.g., $|K|= \polylog(n)$, sampling from $\xi$ reduces to $(1 / \polylog(n))$-approximate CP, that cannot be solved faster than $n^{2- o(1)}$ \cite{rubinstein2018hardness}.
}
CP, and the latter requires quadratic time \cite{williams2018difference}. Intuitively, sampling form $\xi$ requires to find a needle in a haystack.
However, $(1+\varepsilon)$-approximate $\EMD$ is resilient to small perturbation of distances, which we can leverage to obviate this issue.
Consider a \emph{rounding} $\tilde C_{ij}$ of the distances in $C_{ij}$ to the closest power of $1+\varepsilon$, either up or down.
Clearly, $\emd_{C} \approx \emd_{\tilde C}$, so, to approximate $\EMD$, it suffices to fix a rounding $\tilde C_{ij}$ \emph{upfront} and then run our whole algorithm with respect to $\tilde C_{ij}$.
After rounding $C_{ij}$, in what follows we will see that a $(1+\varepsilon)$-approximate CP algorithm is sufficient to sample from $\xi$.

Define the ``$t$-prefix set'' $\cL_t$ as the set of pairs at distance at most $(1+\varepsilon)^t$. Now suppose we had a $(n^{2-\Omega(\phi)})$-time procedure $\allclosepairs(t)$ that returns $\cL_t$ as long as $|\cL_{t+1}| = O(n^{1+\phi})$.

By sampling uniformly random pairs $(i, j) \in [n]\times [n]$ and checking if $(i, j) \in \cL_r$ (for every $r \geq 0$), we can easily find $t$ such that: (i) $|\cL_{t+1}| = O(n^{1+\phi})$

(ii) $\cL_{t+2} \setminus \cL_t$ is large, i.e., $|\cL_{t+2} \setminus \cL_t | = \tilde \Omega(n^{1+\phi})$.

We then define the distances $\tilde C$ as a function of a set $S \subseteq [n] \times [n]$ (to be defined in the following), by rounding down the distances for all pairs $(i,j)$ in $S$, and rounding up every other pair. So long as we have $|S| = n^{2 -\Omega(\phi)}$ we can afford to compute the distances $\tilde C$ upfront and
store them explicitly.

Our sampling algorithm will partition the pairs $A \times B$ in two sets $\cE$ and $\cI$, and develop different procedures to sample from $\xi|_\cE$ and $\xi|_\cI$.
Then, to sample from $\xi$ we just need to estimate the total weight $\sum_{(i, j) \in Z} \exp(|K| / \tilde C_{ij})$ for $Z = \cE, \cI$.

First, we define $\calE$ as the set of pairs handled explicitly $\calE = \cL_t \cup (S\cap \cL_{t+1})$, and define $\calI$ to be the set of pairs handled implicitly, i.e.  $\calI = [n]^2 \setminus \calE$.
Sampling from $ \xi$ restricted to $\calE$ can be done by explicitly computing individual weights $\exp(|K| / C_{ij})$.
To sample from $ \xi$ restricted to $\calI$, we sample a uniform pair  $(i, j) \in \calI$ and then keep it with probability 
\begin{equation}\label{eqn:rejection-sample-tech-overview}
    \exp(K / \tilde C_{ij}) / \exp(K / (1+\varepsilon)^{t+1}),
\end{equation}
and otherwise we reject it. Since all pairs $(i, j) \in \calI$ satisfy $\tilde C_{ij} \geq (1+\varepsilon)^{t+1}$, the quantity (\ref{eqn:rejection-sample-tech-overview}) is a valid probability, and conditioned on not rejecting a pair it results in the correct distribution. 
Moreover,  for at least
 
$|S \cap \cL_{t+2}| + |\cL_{t+1} \setminus S|$ of these
pairs, namely the points for which $\tilde C_{ij} = (1+\eps)^{t+1}$, the quantity (\ref{eqn:rejection-sample-tech-overview}) is in fact equal to $1$, and therefore these pairs are never rejected. 

Thus if we have
\begin{equation} \label{eq:low rejection overhead}
|S \cap \cL_{t+2}| + |\cL_{t+1} \setminus S| = n^{1 +\Omega(\phi)},
\end{equation}
then the above rejection sampling method successfully generates a sample after $n^{1-\Omega(\phi)}$ attempts. So, we can generate the desired $\tilde{O}(n)$ samples from $\xi$ using $n^{2-\Omega(\phi)}$ time.
Finally, assuming that the set $S$ satisfies $|S| = n^{2-\Omega(\phi)}$ and \Cref{eq:low rejection overhead}, we can implement the $\certify(\lambda)$ oracle in $n^{2-\Omega(\phi)}$ time.

\begin{question} \label{que:consistent rounding}
Can we define $S$ (and, thus, the rounding $\tilde C$) in advance so that it satisfies \Cref{eq:low rejection overhead} every time we call this subroutine?    
\end{question}
To understand our solution to Question \ref{que:consistent rounding}, we must first delve into how we lift Restrictions (\ref{item:restriction coefficient}) and (\ref{item:restriction all close pairs}).

\paragraph{Lifting Restriction (\ref{item:restriction coefficient}): Reducing $\lambda$ to $\propto \exp(\sigma \cdot K / C_{ij})$.}
We now show how to reduce the problem of sampling from $\lambda$ as in \Cref{eq:lambda definition} to that of sampling from a distribution $\propto \exp(\sigma \cdot K / C_{ij})$.

Define $\bar \bbeta:= \sum_{\ell =1}^r \balpha_\ell$,  $\bar \balpha := \sum_{\ell =1}^r \bbeta\ell$.
Notice that the coefficient $\bar \balpha_i - \bar\bbeta_j$ is highly structured. We can leverage this structure to reduce to the case of identical coefficients $\bar \balpha_i - \bar\bbeta_j = K$. 
First, we round $\bar \balpha_i - \bar\bbeta_j$ to the closest power of $(1+ \chi)$, for some small\footnote{
Taking $\chi = \frac{1}{\polylog (n)}$ is sufficient to ensure that we sample from a distribution close enough in total variational (TV) distance to $\lambda$.
} $\chi$. Suppose that $\bar\balpha_i$ and $\bar\bbeta_i$ are sorted in increasing order.

Consider the set $Q_t \subseteq [n] \times [n]$ of pairs $(i, j)$ such that $\bar \balpha_i - \bar\bbeta_j$ is rounded to $(1+\chi)^t$.
$Q_t$ is also highly structured: indeed its boundaries are described by \emph{monotonic} sequences $(1, j_1) \dots (n, j_n)$ where $j_1 \leq \dots \leq j_n$.
As a consequence\footnote{
Essentially, this is the same approximation argument used to show that monotonic functions are Riemann-integrable.
}, we can tile each $Q_t$ with small combinatorial rectangles of size $s \times s$, with $s = \sqrt[4]{n}$, leaving only $n^{2 - \Omega(1)}$ uncovered pairs in total.

On each combinatorial rectangle, sampling from $\lambda$ corresponds to sampling from a distribution proportional to $\exp(\sigma \cdot K  / C_{ij})$. 
As for the uncovered pairs, we can sample from them explicitly since there are at most $n^{2 - \Omega(1)}$ many of them.
Then, it is sufficient to estimate the normalization constant (which can be done by sampling):
\[
\sum_{(i, j, \sigma) \in \calC \times \{\pm 1\}} \exp(\sigma \cdot K  / C_{ij})
\]
for each combinatorial rectangle $\calC$. To sample, we first sample a combinatorial rectangle proportionally to its normalization constant, and then sample a pair within the sampled rectangle.
If sampling from a single $s\times s$ combinatorial rectangle requires time $s^{2-\Omega(\phi)}$, then the overall running time of this reduction is $\frac {n^2}{s^2} \cdot s^{2-\Omega(\phi)} = n^{2-\Omega(\phi)}$.

\paragraph{Lifting Restriction (\ref{item:restriction all close pairs}): Retrieving all close pairs via CP.}
Note that the fastest approximate CP algorithms \cite{V15, ACW16} leverage similar algebraic methods to each other.
Essentially, these methods lump together a batch $Y_i \subset Y$ of $n^{\phi}$ points in into individual vector $v_i$; then they show that a single dot product $\langle x, v_i\rangle$ is sufficient to detect if there exists $y \in Y_i$ such that $||x-y||$ is small. If that is the case, all distances in $\{x\} \times Y_i$ are computed.
Interestingly, this method can be tweaked to return all close pairs, in the sense of Restriction (\ref{item:restriction all close pairs}). However, we would like a generic reduction to CP, and not rely on the specific structure of the aforementioned algorithms. 
\begin{question}
Can we transform any algorithm for approximate CP into an algorithm that returns all (approximately) close pairs?
\end{question}
We answer this question affirmatively. Specifically, we now show how to implement $\allclosepairs$ in time $n^{2-\Omega(\phi)}$ by leveraging a $n^{2-\phi}$-time algorithm for $(1+\varepsilon)$-approximate $\cp$.
Our goal is to retrieve all pairs is $\cL_t$, and we can assume that $|\cL_{t+2}| \approx n^{1+\phi}$.
We subsample $X' \subseteq X$ and $Y' \subseteq Y$ with rate $1 / z$ for $z := \sqrt{n^{1+\phi}}$, then compute $\cp(X', Y')$.
A pair $(x, y) \in \cL_t$ is guaranteed to be returned by $\cp(X', Y')$ as long as $(X' \times Y') \cap \cL_{t+1} = \{(x, y)\}$. Condition on the event $x \in X'$ and $y \in Y'$, which happens with probability $z^{-2}$.
Then, assuming $|\cL_{t+1}| \leq |\cL_{t+3}| \approx z^2$, a union bound over all pairs $(c, d) \in \cL_{t+1} \cap (X\setminus \{x\}) \times (Y \setminus \{y\})$ 

shows that, with constant probability, none of these pairs lies in $X' \times Y'$.
Notice that the bound above is possible because the events $(c, d) \in X' \times Y'$ and $(x, y) \in X' \times Y'$ are independent.
For pairs of the form $(x, d)$ (likewise for $(c, y)$), the probability of $(x, d) \in X' \times Y'$ conditioned on $(x, y) \in X' \times Y'$ is $\deg_{\cL_{t+1}}(x) / z$.

Therefore, we need to consider two cases:
\begin{enumerate}
    \item If $\deg_{\cL_{t+1}}(x) < z/3$, then the probability of any colliding pair $(x, d) \in \cL_{t+1} \cap (\{x\} \times Y \setminus \{y\})$, conditioned on $(x, y) \in X' \times Y'$ is at most $1 - \Omega(1)$. 
    Thus, the probability of $\cp(X', Y')$ returning $(x, y)$ is $\Omega(z^{-2})$ and $\tilde O(z^2)$ repetitions suffice to find $(x, y)$.
    \label{item:small deg}
    \item  If $\deg_{\cL_{t+1}}(x) \geq z/3$, then $\cp(X', Y')$ returns $x$ often, so, we shall see that can \emph{detect} $x$ using $\tilde O(z)$ repetitions and checking what fraction of the times a pair $(x, \cdot)$ is returned.
    \label{item:large def}
\end{enumerate}
Either way, $\tilde O(z^2) = \tilde O(n^{1+\phi})$ calls to $\cp(X', Y')$ suffice. By standard concentration $|X'|, |Y'| \approx \frac{n}{z}$, thus $\cp(X', Y')$ runs in time $(\frac n z)^{2-\Omega(\phi)}$, which gives an overall running time of $n^{2-\Omega(\phi)}$. 

To detect $x$ in point \ref{item:large def}, we need to ensure that, if $\deg_{\cL_{t+1}}(x) \geq z/3$, then a pair $(x, \cdot)$ is returned often enough (i.e., with probability $\Omega(z^{-1})$).
The probability that the $\cL_{t+1}$-neighborhood of $x$ intersects $Y'$ is a large constant.
On the other hand, the probability that any not-incident-to-$x$ pair in $\cL_{t+2}$ lies in $X' \times Y'$ is, by union bound over all $\cL_{t+2}$ such pairs, small.
Thus, with constant probability, the first event happens and the second event does not.
These events are independent of $x \in X'$, which happens with probability $1/z$.
Thus, with probability $\Omega(z^{-1})$ both these events happen and $x$ is the only point in $X'$ incident to some edge in $\cL_{t+2} \cap X' \times Y'$, so, $\cp(X', Y')$ returns $(x, \cdot)$.
Then, for each $x$ detected in point (2), we find the $\cL_t$-neighborhood of $x$ brute-force by computing $||x-y||_1$ for all $y \in Y$. This uses at most $n \cdot \frac{2|\cL_{t+1}|}{z / 3} = n^{2-\Omega(\phi)}$ time.
\\

Now that we have lifted both Restriction (\ref{item:restriction coefficient}) and Restriction (\ref{item:restriction all close pairs}), all that remains is to answer \Cref{que:consistent rounding}. Namely, to define a set $S\subset [n] \times [n]$, so as to fix the rounding $\tilde{C}$, so that it satisfies \Cref{eq:low rejection overhead} every time we attempt to sample from a distribution $\lambda$. 

\paragraph{Solving \Cref{que:consistent rounding} via Fixed Randomized Rounding.}
The core challenge, posed in \Cref{que:consistent rounding}, is to define a rounding scheme $\tilde{C}$ (by specifying the set $S$ of pairs to be rounded down) that guarantees \Cref{eq:low rejection overhead} holds every time the sampling scheme is called, so that the rejection sample step is efficient. 
Recall that \Cref{eq:low rejection overhead} requires that for the relevant distance levels $\mathcal{L}_{t+1}, \mathcal{L}_{t+2}$, the set $S$ has a sufficiently large intersection with $\mathcal{L}_{t+2}$. Further recall that $t$ is defined as the level for which $\mathcal{L}_{t+2} \setminus \cL_{t}$ is ``large'' and $\mathcal{L}_{t+1}$ is not too ``large''. 
As discussed earlier in Section \ref{sec:high-level approach}, the specific sets $\mathcal{L}_{t+1}, \mathcal{L}_{t+2}$ depend on the subproblem (i.e., the combinatorial rectangle $X_k \times Y_k$) being considered, which changes throughout the Multiplicative Weights Update (MWU) iterations. Importantly, 
these subproblems are a function of the output of the MWU algorithm on the last iteration, which depend on the prior distribution $\lambda$ and therefore the rounding. Thus the sets $\cL_t, \mathcal{L}_{t+1}, \mathcal{L}_{t+2}$ are not independent of the rounding we choose.

Our solution, introduced conceptually in Section~\ref{sec:high-level approach} and detailed technically in Section \ref{sec:consistent rounding}, is to use a  randomized rounding strategy despite this dependency. We sample a single set $S \subseteq X \times Y$ of size $n^{2-\phi/2}$ uniformly at random at the beginning of the algorithm. Since $\mathcal{L}_{t+2}$ is defined to have size at least $n^{1+\phi}$, we will have $|S \cap \mathcal{L}_{t+2}| = n^{1+\Omega(\phi)}$ with probability $1-\exp(-\Omega(|S|))$ if $S$ were independent from $\mathcal{L}_{t+2}$. Unfortunately, as described above, it is not independent. Our approach to resolve this independence issue is to demonstrate that the combinatorial rectangles on a given step of MWU can be taken to be a deterministic function of the $\tilde{O}(n)$ samples drawn from the distribution $\lambda$ on the last step. Thus, there are at most $2^{\tilde{O}(n)}$ possible rectangles (and thus possible sets $\mathcal{L}_{r}$) on each step. Since we only need to run MWU for $\polylog(n)$ iterations, we can union bound over the set of all $2^{\tilde{O}(n)}$ possible sets $\mathcal{L}_{r}$ that could ever occur as the relevant set in a given subproblem, thereby guaranteeing correctness.


\paragraph{Comparison with Sherman's Framework.}
In an influential paper, Sherman~\cite{sherman2013nearly} presented a framework for efficiently approximating flow problems with high accuracy. Sherman's preconditioning framework is very versatile and has thus been refined to obtain faster parallel algorithm for
shortest path~\cite{L20, andoni2020parallel} (see also the excellent introduction in Li's thesis~\cite{li2021preconditioning}).
In essence, Sherman's framework reduces $(1+\varepsilon)$-approximate uncapacitated minimum cost flow to $\polylog(n)$-approximate $\ell_1$-\emph{oblivious routing}.

Employed off-the-shelf, Sherman's frameworks falls short of providing improvements for $(1+\varepsilon)$-approximate EMD beyond the already discussed approach of building a spanner and solving the corresponding graph problem on the spanner.

Importantly, current implementations of Sherman's framework do not appear to lend themselves to sublinear-time implementations.

To counter this, we design an ad-hoc oblivious embedding (\Cref{sec:oblivious routing}) and give an alternate implementation of Sherman's framework (\Cref{sec:template}), which we show is amenable of sublinear-time implementation in the size of the LP (\Cref{sec:sampling via cp}).
The bulk of our work is to show that such sublinear implementation can be achieved by leveraging the geometry of the space.

\section{Related Work}\label{sec:related}

\paragraph{Other Approaches to Approximating $\EMD$.}

The focus of this paper is on ``high-quality'' $(1+\eps)$ approximations to high-dimensional $\EMD$, where we achieve the best known runtime (Corollary \ref{cor:main}), and for which $\eps$ is taken to be smaller than some constant. For this regime,  there were no known subquadrtic algorithms prior to \cite{AZ23}, which we improve on here. For $C$-approximations, where $C$ is a large constant,  a classic result in metric geometry is that any $n$-point metric space admits a spanner of size $O(n^{1+1/C})$ (which, moreover, can be constructed efficiently in Euclidean space)~\cite{peleg1989graph,HIM12}. Via fast min-cost flow solvers \cite{CKLPPS22}, this yields a $C$-approximation to $\EMD$ in time $n^{1+O(1/C) +o(1)}$.  Additionally, it is known that $\EMD$ over $(\R^d, \ell_1)$ can be embedded into a tree metric with $O(\log n)$ distortion in $\tilde{O}(n)$ time \cite{CJLW22}, resulting in an algorithm for $\EMD$ with the same runtime and approximation. We note that $\EMD$ cannot be computed exactly in high-dimensions in truly-subquadratic time (conditioned on the Orthogonal Vectors Conjecture or SETH), \cite{rohatgi2019conditional}, hence the emphasize on approximations.

Note that for low-dimensional spaces, where $d$ is taken to be a constant, 
it is known that $\EMD$ can be approximated to $(1+\epsilon)$ in near-linear time \cite{sharathkumar2012near,ACRX22}, although incurring exponential factors in the dimension (i.e. $(1/\eps)^d)$.

The closest pair problem itself is central to computational geometry. It is tightly related to nearest neighbor search \cite{IM98,HIM12,AI08,A09}, and indeed computing approximate nearest neighbors (ANN) for all points is sufficient to solve approximate CP. Given that the best $(1+\eps)$-approximate nearest neighbor search algorithms use $n^{2-\Theta(\varepsilon)}$ preproceesing time and $n^{1+\Theta(\eps)}$ query time \cite{AR15}, this yields a CP algorithm with runtime $n^{2-\Theta(\varepsilon)}$. In a suprising breakthrough, Valiant demonstrating that the dependency on single-query ANN can be beat using the polynomial methodm resulting in a $n^{2-\Omega(\sqrt{\eps})}$ algorithm. This was subsequently improved by \cite{ACW16} to the current best known $n^{2-\tilde{\Omega}(\eps^{1/3})}$ runtime.

\paragraph{Previous reductions of geometric problems to ANN and CP.}
Several other works have utilized approximate nearest neighbor (ANN) or closest pair as a subroutine to solve another geometric task. For instance, in \cite{HIM12} they show how one can solve the Minimum Spanning Tree (MST) problem  via calls to a dynamic ANN algorithm (which they implement via LSH).
For low-dimensional space, \cite{agarwal1990euclidean} reduces MST to closest pair, though they incur  incurring an exponential dependence on the dimension $d$. For high-dimensional space, \cite{ACW16} reduces MST to CP, achieving runtimes comparable to their CP algorithm. 
For large constant approximations of $\EMD$, \cite{agarwal2014approximation} leverage a $(1+\varepsilon)$-ANN subroutine with running time $\tau(n, \varepsilon)$ to obtain a $O(1 / \delta^{1+\varepsilon})$-approximation of $\EMD$ in time $\tilde O(n^{1 + \delta} \tau(n, \varepsilon))$. However, the latter approach is now dominated by the spanner with fast min-cost flow solver approach described above.

\section{Preliminaries}

\paragraph{Notation.}
For any integer $n \ge 1$, we write $[n] = \{1, 2, \dots, n\}$, and for two integers $a, b \in \Z$, write $[a..b] = \{a, a+1, \dots, b\}$. For $a, b \in \R$ and $\epsilon \in (0, 1)$, we use the notation $a = (1 \pm \epsilon)b$ to denote the containment of $a \in [(1 - \epsilon)b, (1 + \epsilon)b]$. We will use boldface symbols to represent random variables and functions, and non-boldface symbols for fixed values (potentially realizations of these random variables) for instance $\boldsymbol{f}$ vs.\ $f$. We use $\tilde{O}(\cdot),\tilde{\Omega}(\cdot)$ notation to hide polylog factors in $n,d,\eps$.

\paragraph{Choice of $\ell_p$ norm, and magnitude of the coordinates.}
For the remainder of the paper, we assume our point sets $X,Y$ lie in $(\R^d, \ell_1)$. This is without loss of generality, since for any $p \in [1,2]$ there exists an embedding of $\ell_p^d$ to $\ell_1^{d'}$ with approximately preserves all distances to a factor of $(1+\eps)$, where $d' = O(d \log(1/\eps) / \eps^2)$ \cite{johnson1982embedding}. Naturally, our results equally apply to any underlying metric which can be similarly embedded into $\ell_1$. 

Without loss of generality, we assume our pointsets $X,Y$ are contained in $[1, \Phi]^d$ for some integer $\Phi \geq 0$. For simplicity and to avoid additional notation, we also allow $\Phi$ to be a bound on the aspect ratio of $A,B$, namely, $1 \leq \|x-y\|_1 \leq \Phi$ for any $(x,y) \in X \times Y$ (note that, a priori, the aspect ratio would be at most $\Phi \cdot d$. Note that we can easily remove and match duplicate points in $X,Y$ in $O(nd)$ time in pre-processing, and may therefore assume that $X \cap Y = \emptyset$. In Appendix \ref{sec:poly-aspect}, we will show a reduction to $\Phi = \poly(n,d)$ bounded aspect ratio. Further, we assume withlot loss of generality that $\Phi = \omega(1)$.

\paragraph{Approximate Closest Pair.}   We now formally define the approximate closest pair problem, which is the key problem that we will reduce to. Note that we allow our closest pair oracle to be randomized and succeed with probability at least $2/3$. Of course, this may be boosted to an arbitrary $1-1/\poly(n)$ success probability via independent repetition. 

\begin{definition}[Approximate Closest Pair]\label{def:CP} Fix any $\eps \in (0,1)$. Then the $(1+\varepsilon)$-approximate closest pair (CP) problem as the following. Given $A, B \subseteq (\bbR^d, \ell_1)$ with $|A|, |B| \leq n$ and a parameter $\varepsilon > 0$, return $(a, b) \in A\times B$ such that $||a-b||_1 \leq (1+\varepsilon) \cdot \min_{x \in A, y\in B} ||x-y||_1$. We allow a CP algorithm to be randomized, so long as it succeeds with probability at least $2/3$.

    \end{definition}
\section{Bi-Lipschitz Quadtree Embedding}
\label{sec:oblivious routing}

In this section, we give an approximate tree embedding embedding of our dataset with stronger guarantees than what is possible via standard  probabilistic tree embeddings. Our construction is a post-hoc modification to the probabilistic tree embeddings from \cite{AIK08}.

Specifically,~\cite{AIK08} shows how to construct a probabilistic tree embedding for $(\R^d, \ell_1)$, such that each pairwise distance on the tree is lower bounded with high probability, and each pairwise distance on the tree is upper bounded in \emph{expectation}, by the original distance. While these guarantees are asymmetric in nature, they are sufficient for approximating $\EMD$ up to a poly-logarithmic factor. However, for our purpose we will need a stronger bi-Lipschitz guarantee---namely, all pairwise distances should be preserved up to a poly-logarithmic factor with a tree metric. Unfortunately, for a general point set $X \subset (\R^d, \ell_1)$, a bi-Lipschitz tree embeddings with poly-logarithmic distortion does not exist;\footnote{For example, for large $n$, the $n$-cycle can be approximated by a large circle in $\R^2$, and $n$-cycle metrics are known to require $\Omega(n)$ distortion when mapped into trees~\cite{RR98}.} however, we show how to slightly perturb points in $X$ so that such tree embeddings exist while simultaneously ensuring that the Earth Mover's distance on the perturbed points remains the same up to a $(1+\epsilon)$-factor. Importantly, this perturbation is performed \textit{post-hoc}, after the tree embedding itself has been generated, as it depends on the tree embedding itself.

We consider the following modified notation for expressing $\EMD$ for arbitrary supply and demand vectors over differing metrics.

\begin{definition}
    Let $X = ([n], \sfd)$ be a metric, and let $V \subset \R^n$ be the subspace of vectors $b \in \R^n$ which satisfy $\langle b, \one_n\rangle = 0$. For any $b \in V$, the value $\EMD_X(b) \in \R_{\geq 0}$ is defined as follows: 
    \begin{itemize}
        \item Let $b_{+}, b_{-} \in \R_{\geq 0}^n$ divide positive and negative parts of $b$, so $b_{+} - b_{-} = b$.
        \item We consider the complete bipartite graph $[n] \times [n]$, where the cost of the edge $(i,j) \in [n] \times [n]$ is $\sfd(i,j)$.
        \item A feasible flow $\gamma \in \R^{n\times n}_{\geq 0}$ is one which satisfies $\sum_{j=1}^n \gamma_{ij} = (b_{+})_i$ for all $i \in [n]$ and $\sum_{i=1}^n \gamma_{ij} = (b_{-})_j$ for all $j \in [n]$.
    \end{itemize}
    We let 
    \[ \EMD_{X}(b) = \min_{\substack{\gamma \in \R^{n\times n}_{\geq 0} \\ \text{feasible flow}}} \sum_{i=1}^n \sum_{j=1}^n \gamma_{ij} \cdot \sfd(i, j), \]
    and will often refer to the flow $\gamma$ realizing $\EMD_X(b)$ as the minimizing feasible flow.
\end{definition}

One useful property of $\EMD_{X}(b)$ as taken above is that it forms a norm over the vector space $V$ of vectors $b$ whose coordinates sum to zero (i.e. demand is equal to supply).  
\begin{fact}\label{fact:emd-norm}
     Let $X = ([n], \sfd)$ be a metric space, and let $V = \{b \in \R^d \; | \; \langle b, \one_d \rangle = 0 \}$. Then $\EMD_X(b)$ is a norm over $V$.
\end{fact}
We often consider a size-$n$ subsets $X = \{ x_1, \dots, x_n \} \subset (\R^d, \ell_1)$. In these cases, we implicitly associate $\{x_1, \dots, x_n \}$ with the set $[n]$, and refer to $X$ as the metric with $\sfd(i,j) = \|x_i - x_j\|_1$. We use the following probabilistic tree embedding from~\cite{AIK08}, which embeds subsets $X \subset (\R^d, \ell_1)$ of bounded aspect ratio into a probabilistic (rooted) tree via a collection of nested and randomly shifted grids. This allows to translate between the Earth mover's distance among points in $X$, $\EMD_{X}$, to the Earth mover's distance over a tree metric. Recall that a tree metric is specified by a tree $T$ with weighted edges, and defines the metric on the vertices by letting the distance $\sfd_{T}(i,j)$ be the length of the path between $i$ and $j$ in the tree. 

\begin{lemma}[\cite{AIK08}]\label{lem:aik-tree}
    Fix $X = \{x_1,\dots, x_n\} \subset (\R^d, \ell_1)$ whose non-zero pairwise distances lie in $[1,\Phi]$. There is a distribution $\calT = {\cal T} (X)$ which satisfies the following guarantees:
    \begin{itemize}
        \item Every tree in $\calT$ is rooted and has $n$ leaves which are associated with $x_1, \dots, x_n$.
        \item The depth of every tree $T$ in the support of $\calT$ is $h = O(\log \Phi)$, with the root at depth $0$ and leaves at depth $h$.
        \item Every edge connecting a vertex at depth $\ell$ to a vertex at depth $\ell+1$ has the same weight $\Phi / 2^{\ell}$.
    \end{itemize}
    Furthermore, for any $\delta > 0$, distances in a randomly sampled tree $\bT \sim \calT$ satisfy:
    \begin{itemize}
        \item For any $i,j \in [n]$ and any $\delta \in (0, 1)$, $\sfd_{\bT}(i, j) \geq \|x_i - x_j\|_1 / O(\log(1/\delta))$ with probability at least $1 - \delta$ over $\bT \sim \calT$.
        \item For any $b \in V$, let $\gamma \in \R^{n\times n}_{\geq 0}$ denote the min-cost flow realizing $\EMD_{X}(b)$. Then, 
        \[ \Ex_{\bT}\left[\sum_{i=1}^n \sum_{j=1}^n \gamma_{ij} \cdot \sfd_{\bT}(i,j) \right] \leq O(\log \Phi) \cdot \EMD_{X}(b).\]
    \end{itemize}
\end{lemma}

The main algorithm of this section is in Figure~\ref{fig:embed-and-perturb}, which allows us to use the probabilistic tree to perturb the points. The perturbation ensures that the tree is a bi-Lipschitz embedding of the perturbed point set, and that the Earth mover's distance in the original point set and the Earth mover's distance in the perturbed point set does not change significantly.
\begin{figure}[h]
\begin{framed}
\noindent \textbf{Algorithm} \textsc{Embed-and-Perturb}$(X, \eps)$. 
\begin{itemize}
    \item \textbf{Input}: a set of points $X = \{ x_1, \dots, x_n \} \subset (\R^d, \ell_1)$ whose pairwise distances lie in $\{0\} \cup [1,\Phi]$, as well as an accuracy parameter $\eps \in (0, 1)$.
    \item \textbf{Output}: A tuple $(\bY, \bT)$, where $\bY = \{\by_1, \dots, \by_n\} \subset (\R^{d+d'}, \ell_1)$ with $d' = O(\log n)$, and $\bT$ is a tree metric with $n$ leaves associated with the points in $\bY$.
\end{itemize} 
The algorithm proceeds as follows:
\begin{enumerate}
\item Sample a tree $\bT \sim \calT$ from the distribution specified in Lemma~\ref{lem:aik-tree}.
\item For every node $v \in \bT$ at depth $\ell$, sample $\bz_v \sim [0,\eps \cdot \Phi / 2^{\ell}]^{d'}$.
\item For each $i \in [n]$, let $\bv_0, \dots, \bv_h$ denote the root-to-leaf path in $\bT$, where $\bv_h$ is leaf $i$. Let $\by_i \in \R^{d+d'}$ where the first $d$ coordinates are $x_i$, and the final $d'$ coordinates are set to \smash{$(1/d')\sum_{i=0}^h \bz_{v_i}$}.
\end{enumerate}
\end{framed}
\caption{The algorithm \textsc{Embed-and-Perturb} which receives a set of vectors $X$, and produces a perturbed set of vectors and a bi-Lipschitz tree embedding for the perturbed vectors.}\label{fig:embed-and-perturb}
\end{figure}

\begin{lemma} \label{lem:quadtree existence}
There is a constant $c > 0$ such that for any $X = \{x_1,\dots, x_n\} \subset (\R^d, \ell_1)$ with distances in $[1, \Phi]$, the algorithm \textsc{Embed-and-Perturb}$(X, c\eps / \log \Phi)$ outputs a tuple $(\bT, \bY)$ which, for any $b \in V$, satisfies the following with probability at least $0.9$:
\begin{itemize}
    \item The set $\bY = \{ \by_1, \dots, \by_n \} \subset (\R^{d + O(\log n)}, \ell_1)$ satisfies $\EMD_{X}(b) \leq \EMD_{\bY}(b) \leq (1+\eps) \cdot \EMD_{X}(b)$.
    \item The set $\bY$ is embedded in the tree metric $\bT$ of $n$ leaves at depth $O(\log \Phi)$ which achieves bi-Lipschitz distortion $O(\log n \log \Phi / \eps)$.
\end{itemize}
Furthermore, the algorithm runs in time $O(nd \log n \log \Phi)$.
\end{lemma}

\Cref{lem:quadtree existence} is 
proven in Claims~\ref{lem:tree-bilip} and~\ref{lem:perturbed-emd}. First, that the tree $\bT$ is a bi-Lipschitz embedding of the output vectors $\bY$, and then, that for a fixed vector $b \in V$, the Earth mover's distance $\EMD_{X}(b)$ is approximately the same as $\EMD_{\bY}(b)$. In order to do the analysis, we consider the following three events, where each occurs with high probability and imply the desired guarantees: 
\begin{itemize}
    \item $\calbE_1$: This event depends on the draw of $\bT \sim \calT$ from Lemma~\ref{lem:aik-tree}. It occurs whenever every $i, j \in [n]$ satisfies $\|x_i - x_j\|_1 / O(\log n) \leq \sfd_{\bT}(i,j)$. Setting $\delta = 1/n^{10}$ in Lemma~\ref{lem:aik-tree} and by union bound, $\calbE_1$ occurs with probability $1 - o(1)$. 
    \item $\calbE_2$: We consider a fixed vector $b \in V$, and we let $\gamma \in \R^{n\times n}_{\geq 0}$ be the flow which realizes $\EMD_{X}(b)$. This event depends on the draw of $\bT \sim \calT$ from Lemma~\ref{lem:aik-tree}, and occurs whenever
    \[ \sum_{i=1}^n \sum_{j=1}^n \gamma_{ij} \cdot \sfd_{\bT}(i,j) \leq O\left(\log \Phi\right) \cdot \EMD_{X}(b). \]
    From Lemma~\ref{lem:aik-tree} and Markov's inequality, this event occurs with probability at least $0.99$.
    \item $\calbE_3$: Fix a draw of $\bT$, and this event depends on the draws $\bz_{\bv}$ from Line~2 of Figure~\ref{fig:embed-and-perturb}. It occurs whenever the following holds for every $i, j \in [n]$: we let $\bv_0, \dots, \bv_h$ and $\bu_0, \dots ,\bu_h$ be the root-to-leaf paths to leaves $i$ and $j$ in $\bT$, respectively; then, 
    \[ \left\| \frac{1}{d'} \sum_{\ell=0}^h (\bz_{\bv_{\ell}} - \bz_{\bu_{\ell}} ) \right\|_1 = \Theta(\eps) \cdot \sfd_{\bT}(i,j).\]
    We prove that event $\calbE_3$ occurs with high probability below.
\end{itemize}

\begin{claim}
    Fix any tree $\bT$ from Lemma~\ref{lem:aik-tree} and any $i, j \in [n]$. Let $\bv_0, \dots, \bv_h$ and $\bu_0,\dots, \bu_h$ be the root-to-leaf paths to $i$ and $j$ in $\bT$, respectively. Then, 
    \begin{align*}
    \left\| \frac{1}{d'} \sum_{\ell=0}^h (\bz_{\bv_{\ell}} - \bz_{\bu_{\ell}}) \right\|_1 = \Theta(\eps) \cdot \sfd_{\bT}(i,j)
    \end{align*}
    with probability at least $1-1/n^{10}$ over the draw of $\bz_{\bv_{\ell}}, \bz_{\bu_{\ell}} \sim [0, \eps \cdot \Phi / 2^{\ell}]^{d'}$ for $\ell \in \{0, \dots, h\}$.
\end{claim}

\begin{proof}
    For any $k \in [d']$, we let $\bS_k \in \R$ denote the $k$-th coordinate of the vector $(1/d')\sum_{\ell=0}^h (\bz_{\bv_{\ell}} - \bz_{\bu_{\ell}})$. We may equivalently write the random variable
    \begin{align*}
        \bS_k = \frac{1}{d'}\sum_{\ell=0}^h (\bz_{\bv_{\ell}} - \bz_{\bu_{\ell}})_k = \frac{1}{d'} \sum_{\ell=0}^h \ind\{\bv_{\ell} \neq \bu_{\ell} \} \cdot \frac{\eps \cdot \Phi}{2^{\ell}}\left(\ba_k - \bb_k  \right),
    \end{align*}
    where $\ba_k, \bb_k \sim [0, 1]$. Since the random variable $\ba_k - \bb_k$ is symmetric about the origin, the Khintchine inequality implies that 
    \begin{align*}
        \Ex\left[ |\bS_k| \right] = \Theta(1/d') \cdot \left( \sum_{\ell=0}^h \ind\{ \bv_{\ell} \neq \bu_{\ell} \} \cdot \frac{\eps^2 \cdot \Phi^2}{2^{2\ell}}\right)^{1/2} = \Theta(\eps/d') \cdot \sfd_{\bT}(i,j),
    \end{align*}
    where the last equality comes from the draw of $\bT$ in Lemma~\ref{lem:aik-tree}, since $\sfd_{\bT}(i, j)$ is exactly $2 \sum_{\ell=0}^h \ind\{ \bv_{\ell} \neq \bu_{\ell} \} \cdot \Phi / 2^{\ell}$. Similarly, $|\bS_k|$ is bounded by $\Theta(\eps/d') \cdot \sfd_{\bT}(i,j)$ as well, the event follows from standard concentration inequalities on sums of $d' = \Omega(\log n)$ independent and bounded random variables.  
\end{proof}

\begin{claim}\label{lem:tree-bilip}
Execute \textsc{Embed-and-Perturb}$(X, \eps)$ and assume events $\calbE_1$ and $\calbE_3$ hold. Then, every $i, j \in [n]$ satisfies
\[ \eps \cdot \|\by_i - \by_j \|_1 \leq \sfd_{\bT}(i,j) \leq O\left(\log n\right) \cdot \|\by_i - \by_j\|_1. \]
\end{claim}

\begin{proof}
We lower bound
 \begin{align*}
     \| \by_i - \by_j \|_1 \geq \left\| \frac{1}{d'}\sum_{\ell=0}^h (\bz_{\bv_{\ell}} - \bz_{\bu_{\ell}} )\right\|_1 = \Omega(\eps)\cdot \sfd_{\bT}(i, j),
 \end{align*}
 from $\calbE_3$, and furthermore, we can upper bound using both $\calbE_1$ and $\calbE_3$
 \begin{align*}
     \|\by_i - \by_j\|_1 = \|x_i - x_j\|_1 + \left\| \frac{1}{d'}\sum_{\ell=0}^h (\bz_{\bv_{\ell}} - \bz_{\bu_{\ell}} )\right\|_1  \leq O(\log(n)) \cdot \sfd_{\bT}(i,j) + O(\eps) \cdot \sfd_{\bT}(i,j),
 \end{align*}
 which is at most $O(\log n) \cdot \sfd_{\bT}(i, j)$, since $\eps \in (0, 1)$.
\end{proof}

\begin{claim}\label{lem:perturbed-emd}
For any vector $b \in V$, execute $\textsc{Embed-and-Perturb}(X, \eps)$ and assume $\calbE_2$ and $\calbE_3$ hold. Then,
\[ \EMD_{X}(b) \leq \EMD_{\bY}(b) \leq (1+O(\eps \log \Phi)) \cdot \EMD_{X}(b). \]
\end{claim}

\begin{proof}
    The upper bound, $\EMD_X(b) \leq \EMD_{\bY}(b)$, occurs with probability $1$ because every $i,j \in [n]$ always satisfy $\|\by_i - \by_j\|_1 \geq \|x_i - x_j\|_1$. This means that any feasible flow will have its cost be larger in $\bY$ than in $X$. In order to prove the upper bound, let $\gamma \in \R^{n\times n}_{\geq 0}$ denote the feasible flow which realizes $\EMD_{X}(b)$, so
    \[ \EMD_{X}(b) = \sum_{i=1}^n \sum_{j=1}^n \gamma_{ij} \cdot \|x_i - x_j\|_1. \]
    We can upper bound the cost of this flow by first applying the conditions of $\calbE_3$, followed by the conditions of $\calbE_2$:
    \begin{align*}
        \sum_{i=1}^n \sum_{j=1}^n \gamma_{ij} \cdot \|\by_i - \by_j\|_1 &\leq \sum_{i=1}^n \sum_{j=1}^n \gamma_{ij} \left(\|x_i - x_j\|_1 + \Theta(\eps) \cdot \sfd_{\bT}(i, j) \right) \\
                &= \EMD_{X}(b) + \Theta(\eps)\cdot \sum_{i=1}^n \sum_{j=1}^n \gamma_{ij} \cdot \sfd_{\bT}(i,j) \\
                &\leq (1+O(\eps \log \Phi)) \cdot \EMD_{X}(b).
    \end{align*}    
\end{proof}

\paragraph{Reduction to Polynomial Aspect Ratio.} We show that it suffices to consider a bounded aspect ratio $\Phi = \poly(nd\varepsilon^{-1})$. 

\begin{lemma}\label{lem:aspect-ratio}
    There exists a randomized algorithm which runs in $O(nd + n \log n)$ time, succeeds with probability $0.9$, and has the following guarantees:
    \begin{itemize}
        \item \emph{\textbf{Input}}: A set $X = \{ x_1,\dots, x_n\} \subset (\R^d, \ell_1)$ and an integer vector $b \in V$.
        \item \emph{\textbf{Output}}: Collections of points $\bY_1,\dots, \bY_t \subset [1, \Phi]^{d + O(\log n)}$ and a corresponding set of supply/demand integer vectors $b_1,\dots, b_t \in V$.
    \end{itemize}
     Moreover, $\Phi = \poly(nd\varepsilon^{-1})$, and 
    \[ \left| \EMD_{X}(b) - \sum_{i=1}^t \EMD_{\bY_i}(b_i) \right| \leq \eps \cdot \EMD_X(b). \]
\end{lemma}

The proof of this lemma is in Appendix~\ref{sec:poly-aspect}, and proceeds along the following lines. First, we show that we may obtain a polynomial approximation to $\EMD_{X}(b)$ in near-linear time. The polynomial approximation is enough for us to randomly partition the instance into parts of small diameter, such that it suffices to independently solve each part. Finally, we add a small amount of noise to each point in each part, which ensures that the minimum pairwise distance is lower bounded, without significantly affecting the value of the Earth Mover's Distance.
\section{Computing EMD via Multiplicative Weights Update}
\label{sec:template}

\newcommand{\sfU}{\mathsf{U}}
\newcommand{\gammagap}{\gamma_{\textsf{gap}}}

In this section, we present a template of our algorithm for computing $\EMD_{\ell_1}(X, Y)$, for input multi-sets $X = \{ x_1, \dots, x_n \}$ and $Y = \{ y_1, \dots, y_n\}$ from $\R^d$. We seek an algorithm which obtains a $(1+\eps)$-approximation for $\eps \in (0, 1)$. Following Section~\ref{sec:oblivious routing}, we may assume the following facts about our input:

\begin{itemize}
\item The points $x_1,\dots, x_n, y_1,\dots, y_n \in \R^{d}$ all have non-zero pairwise distances in $[1, \Phi]$, where we may assume $\Phi$ is a power of $2$ and at most $\poly(nd\varepsilon^{-1})$. 
\item There is a rooted tree $T$ of depth $h = O(\log \Phi)$ with root at depth $0$ and $2n$ leaves (associated with each point of $X$ and $Y$) at depth $h$, and every edge from depth $\ell$ to $\ell+1$ has weight $\Phi / 2^{\ell}$ (which will be an integer). 
\item The tree $T$ gives a bounded-distortion embedding for pairs of points in $X$ and $Y$; i.e., for parameters $\sfD_\ell = \eps / \log \Phi$ and $\sfD_u = O(\log n)$, every pair $i, j \in [n]$ satisfies
\begin{align} 
\sfD_{l} \cdot \| x_i - y_j\|_1 \leq \sfd_{T}(i,j) \leq \sfD_u \cdot \|x_i - y_j\|_1. \label{eq:dist-guarantee}
\end{align}
\end{itemize}
Recall that a bottom-up greedy assignment from $X$ to $Y$ in $T$ gives a matching which realizes $\EMD_{T}(X, Y)$. This bottom-up greedy matching may be computed in $O(n \log \Phi)$ time, and because all pairwise distances satisfy (\ref{eq:dist-guarantee}), we find a number $t_0 \in \R_{\geq 0}$ which satisfies, for $\sfD_T = \sfD_u / \sfD_\ell =  O(\eps^{-1} \cdot \log n \cdot \log \Phi)$,
\[ t_0 \leq \EMD_{\ell_1}(X, Y) \leq \sfD_T \cdot t_0. \]
In this section, we show that there exist a randomized algorithm satisfies the following: if $\EMD_{\ell_1}(X, Y) \geq (1+3 \eps) t$, it can produce a ``certificate'' that $\EMD_{\ell_1}(X, Y)$ is larger than $t$ (which is correct except with a negligible probability); if it cannot find such certificate, it outputs ``fail.''  This guarantee is enough to obtain a $(1+O(\eps))$-approximation to $\EMD_{\ell_1}(X, Y)$ up to an additional $O(\eps^{-1} \cdot \log \sfD_T )$ factor in the running time, by starting at $t = t_0$ and outputting the smallest $t \leq \sfD_T t_0$ where the algorithm outputs ``fail.'' First, we establish the notion of the certificate to prove that $\EMD_{\ell_1}(X, Y) \geq t$. 

\begin{definition}[Consistent Rounding of Distances]\label{def:rounded-distances}
Fix $\eps \in (0, 1)$ and the set of points $X = \{ x_1,\dots, x_n\}$ and $Y = \{y_1, \dots, y_n\}$ with distances in $[1, \Phi]$. Consider the two matrices:
\begin{itemize}
\item \emph{\textbf{Rounded Distances.}} The matrix $\psi \in \{ 0, \dots, \log \Phi / \eps \}^{n\times n}$ is such that $i, j \in [n]$ satisfies $(1+\eps)^{\psi_{ij}} \leq \|x_i - y_j\|_1 < (1+\eps)^{\psi_{ij} + 1}$.
\item \emph{\textbf{Up/Down Rounding}}. The boolean matrix $S \in \{0,1\}^{n \times n}$, to be specified, which will decide how to round the distance of pair $(i, j)$.
\end{itemize}
We define $C = C(S)$ to be the $n \times n$ cost matrix of rounded distances, where ${C_{ij} = (1+\eps)^{\psi_{ij} - 2\cdot S_{ij}}}$, and note that
\[
1 \leq \dfrac{\|x_i - y_j\|_1}{C_{ij}} \leq 1+4\eps.
\]
\end{definition}

\begin{definition}[Cost Certificate]
For an $n\times n$ matrix $C$ (from Definition~\ref{def:rounded-distances}), we let $\alpha, \beta \in \R^n$ be two vectors, and for $t \geq t_0$, we let
\[ \Gamma_t = \left\{ (\alpha, \beta) \in \R^{n} \times \R^n  : \frac{1}{t} \cdot \left( \sum_{i=1}^n \alpha_i - \sum_{j=1}^n \beta_j \right) - \max_{i,j} \frac{|\alpha_i - \beta_j|}{C_{ij}} > 0 \right\}. \]
\end{definition}

\begin{claim}\label{cl:cert}
If $\Gamma_t \neq \emptyset$, then $\EMD_{\ell_1}(X, Y) \geq t$.
\end{claim}

\begin{proof}

Suppose $(\alpha, \beta) \in \Gamma_t$ and define 
\[
(\alpha', \beta') := \left(\max_{i, j \in [n]} \frac{|\alpha_i - \beta_j|}{||x_i - y_j||}\right)^{-1} \cdot (\alpha, \beta).
\]
Then, $(\alpha', \beta') \in \Gamma_t$, it is a feasible dual solution, and moreover $\max_{i, j\in [n]} \frac{\alpha'_i - \beta'_j}{||x_i - y_j||_1} = 1$. Thus, $\emd_{\ell_1}(X, Y) \geq \sum_{i\in[n]} \alpha'_i - \beta'_i  > t \cdot \left( \max_{i, j\in [n]} \frac{\alpha'_i - \beta'_j}{C_{i,j}} \right) \geq t$ as desired.
\end{proof}

\paragraph{Algorithmic Plan.} From Claim~\ref{cl:cert}, an algorithm which can find a point $(\alpha, \beta)$ guaranteed to lie in $\Gamma_t$ has a certificate that $\EMD_{\ell_1}(X, Y) \geq t$. Our plan is to: (i) either we find that $\EMD_{\ell_1}(X, Y) \leq (1+3\eps) t$, or (ii) execute a procedure which outputs a pair $(\alpha, \beta)$ such that, except with a small error probability, $(\alpha, \beta) \in \Gamma_t$. A sub-quadratic time algorithm may store a pair $(\alpha, \beta) \in \R^{n} \times \R^{n}$ (as it requires $2n$ numbers to specify). Note, however, that the polytope $\Gamma_t$ is defined by $2n^2$ inequality constraints, which means it is not at all clear how to verify whether $(\alpha, \beta) \in \Gamma_t$ (much less find one). 
We will find such a pair $(\alpha, \beta)$ with a multiplicative weights update rule. The algorithm will be iterative, and each iteration will require us to solve the task in Figure~\ref{fig:cert-task}.

\begin{figure}[H]
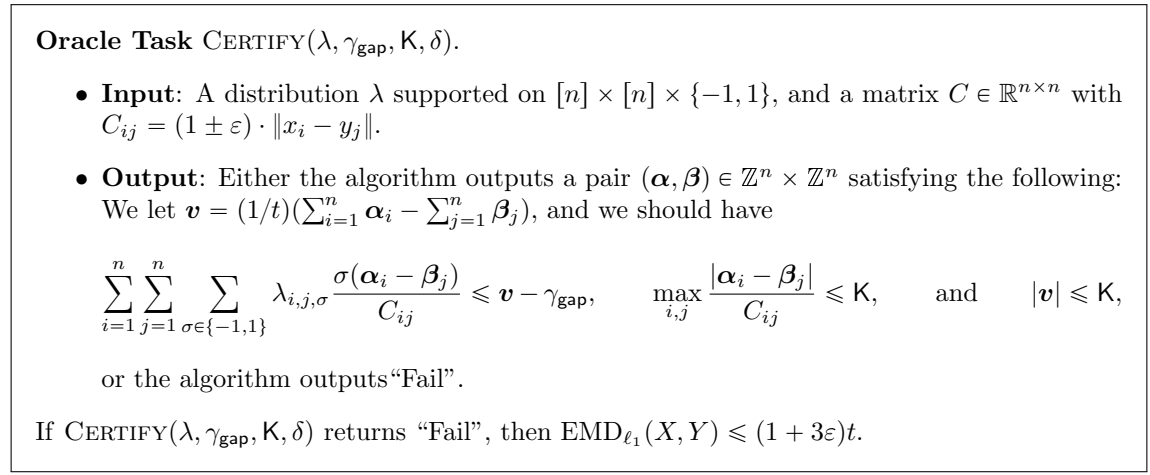

\begin{framed}
\noindent \textbf{Oracle Task} \textsc{Certify}$(\lambda,\gammagap,\sfK, \delta)$. 
\begin{itemize}
    \item \textbf{Input}: A distribution $\lambda$ supported on $[n] \times [n] \times \{-1,1\}$, 
    and a matrix $C \in \bbR^{n\times n}$ with $C_{ij} = (1\pm \varepsilon) \cdot  \|x_i - y_j\|$.
    \item \textbf{Output}: Either the algorithm outputs a pair $(\balpha, \bbeta) \in \Z^n \times \Z^n$ satisfying the following: We let $\bv = (1/t) ( \sum_{i=1}^n \balpha_i - \sum_{j=1}^n \bbeta_j)$, and we should have
  \[ \sum_{i=1}^n \sum_{j=1}^n \sum_{\sigma \in \{-1,1\}} \lambda_{i,j,\sigma} \frac{\sigma(\balpha_i - \bbeta_j)}{C_{ij}} \leq \bv - \gammagap, \qquad \max_{i,j}\dfrac{|\balpha_i - \bbeta_j|}{C_{ij}} \leq \sfK, \qquad \text{and} \qquad |\bv | \leq \sfK, \]
  or the algorithm outputs``Fail''.
\end{itemize}
If \textsc{Certify}$(\lambda,\gammagap,\sfK, \delta)$ returns ``Fail'', then $\emd_{\ell_1}(X, Y) \leq (1+ 3\varepsilon) t$.
\end{framed}
\caption{The certification oracle task for each iteration.}\label{fig:cert-task}
\end{figure}

Lemma \ref{lem:sep} below demonstrates an algorithm that can solve the $ \textsc{Certify}$ task by taking only $\tilde{O}(n/\delta )$ samples. For technical reasons, we will not be able to sample exactly from $\lambda$, but instead will have access to samples from a distribution $\lambda'$ which is very close to $\lambda$ in total variational distance. Namely, we can guarantee $\|\lambda' - \lambda\|_{\tv} \leq 1/\polylog(n)$ (for a sufficiently large polylogarithmic factor). Nevertheless, we show that we can still solve $ \textsc{Certify}$ even given access only to the approximate distribution $\lambda'$.

Note that Lemma \ref{lem:sep} has a small failure probability $\delta$, and requires $\tilde{O}(n/\delta)$ samples to succeed. However, since the procedure will only be called $\polylog(n)$ times, it will suffice to set $\delta = 1/\polylog(n)$, and thus we will require only $\tilde{O}(n)$ samples for every call to \textsc{Certify}. 

\begin{lemma}\label{lem:sep}
For any $\eps, \delta \in (0, 1)$, suppose that 
\[ \gammagap \leq \frac{\eps \cdot \sfD_\ell}{ 2 h} \qquad \text{and}\qquad \sfK \geq \sfD_u \cdot \sfD_T, \]
where $\sfD_T = \sfD_u / \sfD_\ell = O(\varepsilon^{-1} \cdot \log n \cdot \log \Phi)$. Moreover, suppose that we have sample access to a distribution $\lambda'$ such that $\|\lambda' - \lambda\|_{\tv} = o\left( \frac{\varepsilon \delta \sfD_\ell}{h \sfD_u}\right)$. Then there is a randomized algorithm which, with probability $1-\delta$, solves $\textsc{Certify}(\lambda, \gammagap, \sfK, \delta)$ using $\tilde{O}(n/\delta )$ samples from $\lambda'$ and running in time $O(n d \log \Phi)$.

\end{lemma}

\begin{proof}
We begin by assuming we have sample access directly to $\lambda$, and later correct for the fact that we only have sample access to $\lambda'$.
Consider the supply and demand vectors $\mu \in \R^n$ and $\nu \in \R^n$ which are given by:
\[ \mu_i = t \sum_{j=1}^n \sum_{\sigma \in \{-1,1\}} \dfrac{\lambda_{i,j,\sigma} \cdot \sigma}{C_{ij}} \qquad \text{and}\qquad \nu_j = t \sum_{i=1}^n \sum_{\sigma \in \{-1,1\}} \dfrac{\lambda_{i,j,\sigma} \cdot \sigma}{C_{ij}}, \]
where we use $\lambda_{i,j,\sigma} \in [0, 1]$ to be the probability that the tuple $(i,j,\sigma)$ is sampled under $\lambda$. Notice that the flow $\gamma_{ij} = t (\lambda_{i,j,1} - \lambda_{i,j,-1}) / C_{ij}$ is by definition a feasible flow for $\EMD_{\ell_1}(\mu, \nu)$ with cost
\[ \sum_{i=1}^n \sum_{j=1}^n |\gamma_{ij}| \cdot \|x_i - y_j\|_1 \leq t \sum_{i=1}^n \sum_{j=1}^n |\lambda_{i,j,1} - \lambda_{i,j,-1}|(1+\eps) \leq (1+\eps) t.  \] 

Here we have two cases: (i) $\EMD_{\ell_1}(\ind - \mu, \ind - \nu) \geq \varepsilon t$; (ii) $\EMD_{\ell_1}(\ind - \mu, \ind - \nu) \leq \varepsilon t$.
If we are in case $(i)$, by the triangle inequality for the Earth Mover's Distance,  
\[
\EMD_{\ell_1}(\ind,\ind) \leq \EMD_{\ell_1}(\mu,\nu) + \emd_{\ell_1}(\ind - \mu, \ind - \nu) < (1+\varepsilon) t + \varepsilon t = (1+2\varepsilon) t.
\]
Since distance in $T$ can be lower bounded by distances in $\ell_1$ (as per (\ref{eq:dist-guarantee})), in case (ii) we have:
\begin{align*}
\EMD_{T}(\ind-\mu, \ind-\nu) \geq \sfD_\ell \cdot \EMD_{\ell_1}(\ind - \mu, \ind-\nu) \geq \sfD_\ell \cdot \eps t, 
    \end{align*}
and hence there exists a depth of the tree $\ell \in [h]$ which satisfies:
\begin{align} \label{eq:val}
 \sum_{v \in T_{\ell}} \frac{\Phi}{2^{\ell}} \left| |X_v| - |Y_v| - t \sum_{\substack{i \in [n] \\ x_i \in v}} \sum_{j=1}^n \sum_{\sigma \in \{-1,1\}} \dfrac{\lambda_{i,j,\sigma} \cdot \sigma}{C_{ij}} + t \sum_{\substack{j \in [n] \\ y_j \in v}} \sum_{i=1}^n \sum_{\sigma \in \{-1,1\}} \frac{\lambda_{i,j,\sigma} \cdot \sigma}{C_{ij}} \right| \geq \frac{\sfD_{l} \cdot \eps t}{h}. 
 \end{align}
Algorithmically, we iterate over tree depths $\ell \in [h]$ and check whether \Cref{eq:val} holds.
If $\Cref{eq:val}$ never holds, then we return ``Fail''. 
As we proved, if $\Cref{eq:val}$ never holds, then it must be that we are in case (i), thus $\emd_{\ell_1}(\ind, \ind) \leq (1+2\varepsilon) t$, ensuring correctness.
Naively, checking whether \Cref{eq:val} holds required knowledge of the values $\lambda_{i, j, \sigma}$.
Later, we will show that this step can be implemented only using samples from $\lambda$.

Suppose, instead, that we find $\ell \in [h]$ satisfying \Cref{eq:val}.
 We now consider the following algorithm to find a setting of $(\alpha, \beta)$.

\begin{enumerate}
\item Let $(\bi_1, \bj_1, \bsigma_1), \dots, (\bi_s, \bj_s, \bsigma_s) \sim \lambda$ denote $s$ independent samples from the distribution.
    \item For each $v \in T_\ell$ and $k \in [s]$, define $\bZ_v^k =  \left(\ind\left\{ x_{\bi_k} \in v \right\} - \ind\left\{ y_{\bj_k} \in v \right\}  \right) \dfrac{\bsigma_k}{C_{\bi_k \bj_k}}$
\item For each $v \in T_{\ell}$, let $\bw_v \in \{-1,1\}$ be given by:
\begin{align*}
\bw_v = \sign \left( |X_v| - |Y_v| - \frac{t}{s} \sum_{k=1}^s \bZ_v^k \right).
\end{align*} 
\item Finally, for every $v \in T_\ell$ and for every $i, j \in [n]$ with $x_i, y_j \in v$, we set
\[ \balpha_i = \bbeta_j = \frac{\Phi}{2^{\ell}} \cdot \bw_v. \]
\end{enumerate}
For each $v \in T_{\ell}$, let $s_v \in \{-1,1\}$ denote the sign which realizes the absolute value in the left-hand side of (\ref{eq:val}) for $v$.
First, suppose 
if we always had $\bw_v = s_v$. In this case, we would obtain the ``optimal'' values $(\balpha^*, \bbeta^*)$, which satisfy that  for every $i \in [n]$ with $x_i \in v$ and $j \in [n]$ with $y_j \in v$, we have
\[ \balpha_i^* = \bbeta_j^* = s_v \cdot \frac{\Phi}{2^{\ell}}, \]
thus, in this case, we may simplify (\ref{eq:val}) as
\begin{align*}
&\sum_{v \in T_{\ell}} \frac{\Phi}{2^{\ell}} \cdot s_v \left( |X_v| - |Y_v| - t \sum_{\substack{i \in [n] }} \sum_{j=1}^n \sum_{\sigma \in \{-1,1\}} \left(\ind\left\{ x_{i} \in v \right\} - \ind\left\{ y_{j} \in v \right\}  \right) \dfrac{\lambda_{i,j,\sigma} \cdot \sigma}{C_{ij}} \right) \\
&\qquad = \sum_{i=1}^n \balpha_i^* - \sum_{j=1}^n \bbeta_j^* - t \sum_{i=1}^n \sum_{j=1}^n \sum_{\sigma \in \{-1,1\}} \lambda_{i,j,\sigma} \cdot \sigma \cdot \dfrac{\balpha_i^* - \bbeta_j^*}{C_{ij}} \geq \frac{\sfD_\ell \cdot \eps t}{h},
\end{align*}
and in particular, 
\begin{align*}
\sum_{i=1}^n \sum_{j=1}^n \sum_{\sigma \in \{-1,1\}} \lambda_{i, j, \sigma} \cdot \sigma \cdot \dfrac{\balpha_{i}^* - \bbeta_{j}^*}{C_{i j}}    &\leq \frac{1}{t} \left( \sum_{i=1}^n \balpha_{i}^* - \sum_{j=1}^n\bbeta_{j}^*\right) -  \frac{\eps \cdot \sfD_{l}}{h}.\\
& = \bv - \gammagap
\end{align*}
So the idealized values $(\balpha^*,\bbeta^*)$ satisfy the first constraint of the \Cref{lem:sep}. 
Thus, in order for the actual values $(\balpha,\bbeta)$'s we compute to satisfy the first constraint, it will suffices to set the values of $\bw_v$ such that 
\begin{equation}\label{eqn:lem4.4Error}
\begin{split}
    &\left| \sum_{v \in T_{\ell}} \frac{\Phi}{2^{\ell}}  \cdot( \bw_v - s_v)\left( |X_v| - |Y_v| - t \sum_{\substack{i,j \in [n] \\ \sigma \in \{-1,1\}}}(\ind \{ x_{i} \in v \} - \ind\{ y_{j} \in v \} ) \frac{\lambda_{i,j,\sigma} \cdot \sigma}{C_{ij}} \right)  \right|\\
& = \eta \leq \frac{\sfD_\ell \cdot \eps t}{2h}
\end{split}
\end{equation}
We proceed to bound the error $\eta$. First, note the following:
\begin{claim}\label{claim:lem44}
    For any $v \in T_\ell$, and $i,j \in [n]$, if $x_i \in v$ and $y_j \notin v$ (or vice-versa), then $C_{ij} \geq \frac{(1-\eps)\Phi}{  2^\ell\sfD_u } $.
\end{claim}
\begin{proof}
    By the guarantees of the tree embedding \ref{eq:dist-guarantee}, and using that $C_{ij} = (1 \pm \eps) \|x_i - y_j\|_1$, we have that
    \[C_{ij} \geq (1-\eps)\|x_i - y_j\|_1 \geq \frac{(1-\eps)}{\sfD_u }\sfd_{T}(i,j) \geq \frac{(1-\eps)\Phi}{2^\ell \sfD_u} \]
    Where the last inequality holds because $x_i,y_j$ split at or above level $\ell$ in the tree. 
\end{proof}
Further, note that any $i,j$ with $x_i \in v$ and $y_j \in v$ (or $x_i \notin v$ and $y_j \notin v)$ contributes zero to the sum $\sum_{\substack{i,j \in [n] \\ \sigma \in \{-1,1\}}}\left(\ind\left\{ x_{i} \in v \right\} - \ind\left\{ y_{j} \in v \right\}  \right) \lambda_{i,j,\sigma} \cdot \sigma/C_{ij}$,  and thus we can restrict the sum of pairs $i,j$ that are split by a vertex $v$.
To proceed, we define
\begin{equation} \label{eq:error eta v}
\eta_v = \frac{t \Phi}{2^\ell} \left|  \sum_{\substack{i,j \in [n] \\ \sigma \in \{-1,1\}}} \left(\ind\left\{ x_{i} \in v \right\} - \ind\left\{ y_{j} \in v \right\}  \right) \frac{\lambda_{i,j,\sigma} \cdot \sigma}{C_{ij}}  - \frac{1}{s} \sum_{k=1}^s \bZ_v^k\right|
\end{equation}
and note that we can bound $\eta$ via $\eta \leq 2\sum_{v \in T_\ell} \eta_v$, where here we are using the observation that for any vectors $a,b,w \in \R^n$ where $w_i = \text{sign}(b_i)$, we have $|\langle w , a \rangle - \|a\|_1| \leq 2\|a-b\|_1$. Since $\ex{\bZ_v^k} =  \sum_{\substack{i,j \in [n] \\ \sigma \in \{-1,1\}}} \left(\ind\left\{ x_{i} \in v \right\} - \ind\left\{ y_{j} \in v \right\}  \right) \frac{\lambda_{i,j,\sigma} \cdot \sigma}{C_{ij}}$, by Jensen's inequality:
\[  \ex{\eta_v} \leq\frac{t \Phi}{2^\ell}  \cdot \sqrt{\Varx{\frac{1}{s }\sum_{i=1}^s \bZ_v^k} }\]
Using Claim \ref{claim:lem44}, we proceed to bound the above via: 
\begin{align*}
    \Varx{ \bZ_v^k} \leq \sum_{\substack{i,j \in [n] \\ \sigma \in \{-1,1\} \\ |\{x_i,x_j\} \cap v| = 1}} 4 \cdot\frac{\lambda_{i,j,\sigma}}{C_{ij}^2} \leq \left(\frac{2^{\ell+2}\sfD_u }{\Phi}\right)^2   \cdot \sum_{\substack{i,j \in [n] \\ \sigma \in \{-1,1\} \\ |\{x_i,x_j\} \cap v| = 1}}  \lambda_{i,j,\sigma} 
\end{align*}
Thus
\[  \ex{\eta_v} \leq \frac{4\sfD_u t }{\sqrt{s} } \cdot\sqrt{ \sum_{\substack{i,j \in [n] \\ \sigma \in \{-1,1\} \\ |\{x_i,x_j\} \cap v| = 1}}  \lambda_{i,j,\sigma} }\]
We have
\begin{align*}
    \ex{\eta} &\leq \sum_{v \in T_\ell}\ex{\eta_v}\leq\frac{4\sfD_u t }{\sqrt{s} } \cdot \sum_{v \in T_\ell} \sqrt{ \sum_{\substack{i,j \in [n] \\ \sigma \in \{-1,1\} \\ |\{x_i,x_j\} \cap v| = 1}}  \lambda_{i,j,\sigma} } \\
    & \leq \sqrt{2n} \frac{4\sfD_u t }{\sqrt{s} } \cdot \left(\sum_{v \in T_\ell}\left|\sum_{\substack{i,j \in [n] \\ \sigma \in \{-1,1\} \\ |\{x_i,x_j\} \cap v| = 1}}  \lambda_{i,j,\sigma} \right| \right)^{1/2} \leq  \sqrt{4 n} \frac{4\sfD_u t }{\sqrt{s} }
\end{align*}

Where we first used the fact that $\sum_{i=1}^m \sqrt{|a_i|} \leq \sqrt{m} (\sum_{i=1}^n |a_i|)^{1/2}$ for any sequence $a_1,\dots,a_m$, which follows from Cauchy-Schwartz, and we next used that each $\lambda_{i,j,\sigma}$ appears in the sum for at most two vertices $v \in T_\ell$, and lastly that $\lambda$ is a distribution. Therefore, setting $s = O\left(\left(\frac{h \sfD_u}{\eps \sfD_\ell}\right)^2 \frac{n}{\delta^2}\right) = \tilde{O}(n/\delta^2 )$, we obtain $ \ex{\eta}  \leq \frac{\eps \delta \sfD_\ell t }{h 100 }$, and the desired bound on $\eta$ holds with probability at least $1-\delta$ by Markov's inequality. 

\paragraph{Checking \Cref{eq:val}.}
The previous analysis shows that the error $\eta_v$ in \Cref{eq:error eta v} satisfies $\eta_v \leq \frac{\sfD_\ell \cdot \varepsilon t}{2h}$.
Thus, given $\ell \in [h]$, we can check if \Cref{eq:val} holds for $\varepsilon' = 2\varepsilon$, which suffices to obtain
\[
\EMD_{\ell_1}(\ind,\ind) \leq \EMD_{\ell_1}(\mu,\nu) + \emd_{\ell_1}(\ind - \mu, \ind - \nu) < (1+\varepsilon) t + \varepsilon' t \leq (1+3\varepsilon) t,
\]
ensuring correctness in the ``Fail'' case.

\paragraph{Sampling from $\lambda'$ instead of $\lambda$.} 
Next, we claim that the same sampling procedure works when our samples $(\bi_1, \bj_1, \bsigma_1),\dots,(\bi_s, \bj_s, \bsigma_s)$ are drawn from $\lambda'$ instead of $\lambda$. This modifies the distribution of the variables $\bZ_v^k$ accordingly. Call the modified variables $\hat{\bZ}_v^k$. To analyze the modified variables, note that we can bound the error of $\eta$ coming from the new variables $\hat{\bZ}_v^k$ via
\begin{equation}\label{eqn:lem54-final}
    \begin{split}
        \ex{\eta} \leq \sum_v \ex{\eta_v} \leq &\frac{t \Phi}{2^\ell} \left|  \sum_{\substack{i,j \in [n] \\ \sigma \in \{-1,1\}}} \left(\ind\left\{ x_{i} \in v \right\} - \ind\left\{ y_{j} \in v \right\}  \right) \frac{\lambda_{i,j,\sigma}' \cdot \sigma}{C_{ij}}  - \frac{1}{s} \sum_{k=1}^s \hat{\bZ}_v^k\right| \\
       & + \frac{t \Phi}{2^\ell} \left| \sum_{\substack{i,j \in [n] \\ \sigma \in \{-1,1\}}} \left(\ind\left\{ x_{i} \in v \right\} - \ind\left\{ y_{j} \in v \right\}  \right) \frac{|\lambda_{i,j,\sigma}'  - \lambda_{i,j,\sigma}|\cdot \sigma}{C_{ij}} \right| \\
       \leq  & \frac{t \Phi}{2^\ell} \sum_v \left|  \sum_{\substack{i,j \in [n] \\ \sigma \in \{-1,1\}}} \left(\ind\left\{ x_{i} \in v \right\} - \ind\left\{ y_{j} \in v \right\}  \right) \frac{\lambda_{i,j,\sigma}' \cdot \sigma}{C_{ij}}  - \frac{1}{s} \sum_{k=1}^s \hat{\bZ}_v^k\right| +  8 \cdot t \cdot \sfD_u \|\lambda-\lambda'\|_{\tv} 
    \end{split}
\end{equation}
Where in the last line, via Claim \ref{eqn:lem4.4Error} we used that $C_{i,j} \geq \frac{(1-\eps)\Phi}{2^\ell \sfD_u}$ for every $(i,j,v)$ where $\ind\left\{ x_{i} \in v \right\} - \ind\left\{ y_{j} \in v \right\}   \neq 0$. We also used that each term $|\lambda_{i,j,\sigma}'  - \lambda_{i,j,\sigma}|$ appears exactly twice in the sum over all vertices $v \in T_\ell$, allowing us to bound the contribution of those terms by $2 \cdot\|\lambda-\lambda'\|_{\tv} $. 
Now note that the first term on the last line of (\ref{eqn:lem54-final}) can be bounded by $\frac{t \Phi}{2^\ell} \sum_v \cdot \sqrt{\Varx{\frac{1}{s }\sum_{i=1}^s \hat{\bZ}_v^k} }$, and thus by $\sqrt{4 n} \frac{4\sfD_u t }{\sqrt{s} }$ via the same sequence of inequalities as above. Using that $\|\lambda' - \lambda\|_{\tv} = o\left( \frac{\varepsilon  \sfD_\ell}{h \sfD_u}\right)$,  we obtain $ \ex{\eta}  \leq \frac{\eps \sfD_\ell t }{h 50 \delta}$ as before, which is only a factor of $2$ larger than the earlier computation using $\lambda$, and the statement again follows by Markov's inequality.

Finally, for the last two constraints, note that $|\balpha_i - \bbeta_j| \leq d_T(i, j) \leq \sfD_u \|x_i - y_j\|_1 \leq \sfD_u (1+\eps) C_{ij}$, and therefore, we have $\sum_{i=1}^n \balpha_i- \sum_{j=1}^n \bbeta_j \leq \sfD_u \cdot \EMD_{\ell_1}(X, Y) \leq \sfD_u \sfD_T t$, where the last inequality holds because we know $t_0 \leq t \leq \sfD_T \cdot t_0$ and moreover $\EMD_{\ell_1}(X,Y) \leq \sfD_T \cdot t_0$. Thus, $(\balpha,\bbeta)$ satisfy the desired constraints.

\end{proof}

\subsection{Multiplicative Weights Update from a Bounded Class of Distributions}

As we explain now, the class of distributions $\lambda$ which we will consider will be parametrized over dual variables $(\alpha, \beta)$, which are the values which get updated as the algorithm proceeds and evolve the distribution $\lambda$. In order to efficiently sample from $\lambda$, it will be important for us to appropriately round to an accuracy $\chi$ and maintain the signs of the dual variables $(\alpha, \beta)$.  
\begin{definition}[Rounded Dual Variables]\label{def:rounded-duals}
Fix $\chi \in (0, 1)$ and a pair $(\alpha, \beta) \in \Z^n \times \Z^n$. 
\begin{itemize}
\item \emph{\textbf{Rounded Duals}}. $D$ is the $n \times n$ matrix with entries in $\{ 0 \} \cup \{ (1+\chi)^h : h \in \Z_{\geq 0} \}$ satisfying
\[ D_{ij} \leq |\alpha_i - \beta_j| \leq (1+\chi) \cdot D_{ij}. \]
\item \emph{\textbf{Sign Pattern}}. $P$ is the $n\times n$ matrix with $P_{ij} = \sign(\alpha_i - \beta_j)$.
\end{itemize}
\end{definition}
We now define a bounded class of distributions $\lambda$ which depend on $C, D$ and $P$. In the remainder of this section, we first define the class of distributions $\lambda$. Then, we show that a multiplicative weights update rule will produce a sequence of distributions $\lambda$ in this class. Furthermore, iteratively solving the certification task on these distributions is enough to find a pair $(\alpha, \beta) \in \Gamma_t$. 

\begin{definition}[Class of Distributions]\label{def:disc-dist}
Fix $\eta \in (0, 1)$ and $S \in \{0, 1\}^{n\times n}$. Define matrices $D, P$ 
as in Definition~\ref{def:rounded-duals} and matrix $C = C(S)$ as in \Cref{def:rounded-distances}.
Define weights $w_{i,j,\sigma} \geq 0$ for $i, j \in [n]$ and $\sigma \in \{-1,1\}$:
\[ w_{i,j,\sigma} =  \exp(\eta \cdot (\sigma \cdot P_{ij}) \cdot D_{ij} / C_{ij}). \]
The distribution $\lambda(\eta, w) = \lambda(\eta,C,D,P)$ is that which samples with probability proportionally to $w$, and $\calD(\eta)$ is the class of distributions specified by $\eta, C, D$ and $P$. 
\end{definition}

In the remainder of the section, we show the multiplicative weights update (MWU) algorithm 
equipped with an oracle for
$\textsc{Certify}(\lambda, \gammagap, \sfK, \delta)$ on distributions $\lambda \in \calD$ can be used to find a certificate $(\alpha, \beta) \in \Gamma_t$. 
In particular, we analyze the iterative procedure in Figure~\ref{fig:mwu}.

Thanks to \Cref{lem:sep}, if we can design an efficient algorithm to sample from distributions in $\calD$, then we can implement the oracle for $\textsc{Certify}(\lambda, \gammagap, \sfK, \delta)$ efficiently, which in turn makes the MWU algorithm efficient.

\begin{figure}[H]
\begin{framed}
\noindent \textbf{Multiplicative Weights Update.} There is a fixed matrix of rounded distances $C$ satisfying $C_{ij} = (1\pm \varepsilon) \cdot \|x_i - y_j\|$.
\begin{enumerate}
    \item Initialize $(\alpha^1, \beta^1) \in \Z^n \times \Z^n$ to all zero. This specifies $v^{1} = 0$, matrix $D^{1}$ to all zero, and $P^1$ to all one, and weights $w_{i,j,1}^{1} = w_{i,j,-1}^1= 1$.
    \item Repeat for rounds $r = 1, \dots, \sfR-1$: 
    \begin{itemize}
    \item Let $\lambda^{r} = \lambda(\eta, w^{r}) = \lambda(\eta, C, D^r, P^r) $ as in \Cref{def:disc-dist}.
    \item If $\textsc{Certify}(\lambda^{r}, \gammagap, \sfK, \delta)$ outputs ``Fail'', return ``Fail''.
    \item Else, let $(\tilde{\alpha}^r, \tilde{\beta}^r)$ 
    be the output of $\textsc{Certify}(\lambda^{r}, \gammagap, \sfK, \delta)$.
    \item Update $(\alpha^{r+1}, \beta^{r+1}) = (\alpha^{r}, \beta^{r}) + (\tilde{\alpha}^r, \tilde{\beta}^r)$. This specifies $v^{r+1} = v^r + \tilde{v}^r$, matrices $D^{r+1}$ and $P^{r+1}$, and weights $w^{r+1}$.
    \end{itemize}
    \item Output $1/ \sfR \cdot (\alpha^{\sfR}, \beta^{\sfR})$ \label{item:last line}
\end{enumerate} 
If this algorithm returns ``Fail'', then $\emd(X, Y) \leq (1+3\varepsilon) t$.
Else, if the algorithm returns $(\alpha, \beta)$, then $(\alpha, \beta) \in \Gamma_t$.
\end{framed}
\caption{The Multiplicative Weights Update Rule.}\label{fig:mwu}
\end{figure}

\begin{lemma}\label{lem:mwu}
Run the multiplicative weights algorithm in Figure~\ref{fig:mwu} with parameters:
\begin{align*}
\eta \leq \frac{\gammagap}{100 \cdot \sfK^2} \qquad \sfR \geq \frac{\ln(2n^2)}{100 \cdot \eta \gammagap} \qquad \text{and}\qquad \chi \leq \frac{\gammagap}{100 \cdot \sfR \sfK}.
\end{align*} 
Suppose it returns $(\alpha, \beta)$ at line \ref{item:last line}.
Then, $(\alpha, \beta) \in \Gamma_t$.
\end{lemma}

\begin{observation} \label{obs:correctness of mwu}
Assume correctness of \textsc{Certify}.
Thanks to \Cref{lem:sep}, if the algorithm in \Cref{fig:mwu} returns ``Fail'', then $\emd(X, Y) \leq (1+3\varepsilon) t$.
Thanks to \Cref{lem:mwu}, if if the algorithm  \Cref{fig:mwu} returns $(\alpha, \beta)$, then $(\alpha, \beta) \in \Gamma_t$, and $\emd_{\ell_1}(X, Y) \leq t$ by \Cref{cl:cert}.
\end{observation}

\subsection{Proof of Lemma~\ref{lem:mwu}}

We first make the following observations, which dictate how the weights $w^r$ evolve in the multiplicative weights update procedure. 
The goal will be to upper bound the sum of all weights in the algorithm, so we let $W^r = \sum_{i=1}^n \sum_{j=1}^n \sum_{\sigma \in \{-1,1\}} w^r_{i,j,\sigma}$.

\begin{observation}\label{obs:main}
For any $i, j \in [n]$ and $\sigma \in \{-1,1\}$. Then, 
\begin{align*}
w_{i,j,\sigma}^{r+1} &\leq w_{i,j,\sigma}^{r} \cdot \exp\left( \eta \chi \cdot 2(r+1) \sfK \right)  \cdot \exp\left( \eta \sigma (\tilde{\alpha}_i^r - \tilde{\beta}_j^r) / C_{ij}\right). 
\end{align*}
\end{observation}

\begin{proof}[Proof of Observation~\ref{obs:main}]
Here, we expand the definition of $w_{i,j,\sigma}^{r+1}$, and we use the fact $P_{ij}^{r+1} D_{ij}^{r+1} \leq \alpha_i^{r+1} - \beta_j^{r+1} + \chi D_{ij}^{r+1}$.
\begin{align*}
w_{i,j,\sigma}^{r+1} = \exp\left( \eta (\sigma \cdot P_{ij}^{r+1}) D^{r+1}_{ij} / C_{ij}\right) &\leq \exp\left( \eta \sigma (\alpha_i^{r+1} - \beta_j^{r+1}) / C_{ij} \right) \cdot \exp\left(\eta \chi D_{ij}^{r+1} / C_{ij} \right).
\end{align*}
We may now expand the first exponential in the right-hand side, and have:
\begin{align*}
	\exp\left( \eta \sigma (\alpha_i^{r+1} - \beta_j^{r+1}) / C_{ij}\right) &= \exp\left( \eta \sigma (\alpha_{i}^{r} - \beta_j^r) / C_{ij} \right) \cdot \exp\left( \eta \sigma (\tilde{\alpha}^r_i - \tilde{\beta}^r_j) / C_{ij} \right) \\
	&\leq w_{i,j,\sigma}^{r} \cdot \exp\left( \eta \chi D_{ij}^r / C_{ij}\right) \cdot \exp\left( \eta \sigma (\tilde{\alpha}_i^r - \tilde{\beta}_j^r) / C_{ij} \right),
\end{align*}
where the last inequality also uses the fact $P_{ij}^{r} D_{ij}^r$ is at least $\alpha_i^r - \beta_j^r + \chi D_{ij}^{r}$. The final inequality follows from combining both of the above, and noting the fact that both $D_{ij}^{r} / C_{ij}$ and $D_{ij}^{r+1} / C_{ij}$ are at most $(r+1) \sfK$ (since they are at worst $r+1$ sums of terms which are smaller than $\sfK$).
\end{proof}

We may now use the choice of $(\tilde{\alpha}^{r}, \tilde{\beta}^r)$ being the solution to $\textsc{Certify}(\lambda^{r}, \gammagap, \sfK, \delta)$ in the analysis of how the weights evolve. The following lemma follows the traditional analysis of the multiplicative weights update method up to the small error incurred from rounding the dual variables.
\begin{observation}\label{obs:weight-increase}
If $(\tilde{\alpha}^{r}, \tilde{\beta}^r)$ satisfies the guarantees of $\textsc{Certify}(\lambda^{r}, \gammagap, \sfK, \delta)$ (in Figure~\ref{fig:cert-task}) and $\eta \leq 1/(100\sfK)$. Then,
\begin{align*}
W^{r+1} &\leq \exp\left( \eta \chi \cdot r \sfK + \tilde{v}^r - \gammagap + O(\eta^2 \sfK^2) \right) \cdot W^r
\end{align*}
\end{observation}

\begin{proof}
Here, we will use the guarantees of Figure~\ref{fig:cert-task}. In particular, we may take second-degree Taylor expansion of the exponential function, and have
\begin{align*}
& \Ex_{(\bi,\bj,\bsigma) \sim \lambda^r}\left[ \exp\left( \eta \cdot \left( \bsigma (\tilde{\alpha}_{\bi}^r - \tilde{\beta}_{\bj}^r) / C_{\bi \bj} - \tilde{v}^r + \gammagap\right)\right)\right] \\
&\qquad\qquad\qquad \leq 1 + \eta \left(\Ex_{(\bi, \bj,\bsigma)\sim \lambda^r}\left[ \frac{\bsigma (\tilde{\alpha}_{\bi}^r - \tilde{\beta}_{\bj}^r)}{C_{\bi\bj}} \right] - \tilde{v}^r + \gammagap \right) + O(\eta^2 \sfK^2) \\
 &\qquad\qquad\qquad\leq 1 + O(\eta^2 \sfK^2) \leq \exp\left( O(\eta^2 \sfK^2) \right).
\end{align*} 
We apply Observation~\ref{obs:main}
\begin{align*}
W^{r+1} &= \sum_{i=1}^n \sum_{j=1}^n \sum_{\sigma} w_{i,j,\sigma}^{r+1} \leq \sum_{i=1}^n \sum_{j=1}^n \sum_{\sigma} w_{i,j,\sigma}^{r} \cdot \exp\left( \eta \chi \cdot 2(r+1) \sfK \right)  \cdot \exp\left( \eta \sigma (\tilde{\alpha}_i^r - \tilde{\beta}_j^r) / C_{ij}\right) \\
		&\leq W^r \cdot \exp\left( \eta \cdot \left( \chi \cdot 2(r+1) \sfK+ \tilde{v}^r - \gammagap\right)\right) \cdot \Ex_{(\bi,\bj,\bsigma) \sim \lambda^r}\left[ \exp\left( \eta \cdot \left( \bsigma(\tilde{\alpha}_{\bi}^r - \tilde{\beta}_{\bj}^{r}) / C_{\bi\bj} - \tilde{v}^r + \gammagap\right)\right) \right] \\
		&\leq W^r \cdot \exp\left(\eta \cdot \left( \chi \cdot 2(r+1) \sfK + \tilde{v}^r - \gammagap + O(\eta^2 \sfK^2) \right)\right).
\end{align*}
\end{proof}

From here, we may conclude the proof of Lemma~\ref{lem:mwu}. In particular, after all $\sfR$ rounds, we have:
\begin{align*}
W^{\sfR} \leq  2n^2 \cdot \left( \eta \chi \cdot 2\sfR^2 \sfK + \eta \sum_{r=1}^{\sfR-1}\tilde{v}^r - \eta (\sfR - 1) \gammagap + O(\sfR \eta^2 \sfK^2) \right),
\end{align*}
and this implies every $i,j \in [n]$ satisfies
\begin{align*}
\frac{|\alpha_i - \beta_j|}{C_{ij}} = \frac{1}{\sfR} \cdot \frac{|\alpha_i^\sfR - \beta_j^{\sfR}|}{C_{ij}} \leq \frac{1}{\eta \cdot \sfR} \ln\left( w_{i,j,1}^{\sfR} + w_{i,j, -1}^{\sfR} \right) \leq \frac{\ln(2n^2)}{\eta \cdot \sfR} + \chi \sfR \sfK + v - \gammagap/2 + O(\eta \sfK^2).
\end{align*}
The parameter settings of Lemma~\ref{lem:mwu} implies 
\[ \max_{i,j} \frac{|\alpha_i - \beta_j|}{C_{ij}} \leq v - \gammagap /10 \leq v, \]
which means $(\alpha, \beta) \in \Gamma_t$.

\subsection{Proof of the Main Theorem}
We now state our main theorem, which provides a $(1+\eps)$ approximation. Note that, by standard independent repetition and outputting the median of estimates, we can boost the failure probability to $1-1/\poly(n)$ with a $O(\log n)$ factor increase in the runtime. 
\begin{theorem}[Main Theorem, formal] \label{thm:main theorem formal}
Suppose that there exists an algorithm for $(1+\varepsilon)$-approximate closest pair (Definition \ref{def:CP}) on $(\mathbb{R}^d, \ell_1)$ with running time $T_{CP}(n, \varepsilon) = n^{2-\phi(\varepsilon)} + O(nd)$ and success probability at least $2/3$.
Then, there exists an algorithm that computes a $(1+O(\varepsilon))$-approximation to $\EMD(X, Y)$ over $(\mathbb{R}^d, \ell_p)$ for any $p \in [1,2]$ in time $\tilde O(n^{2-\Omega(\phi(\varepsilon))} + nd)$
and success probability at least $2/3$.
\end{theorem}
\begin{proof}
By \Cref{lem:aspect-ratio}, we can assume that our input is two size-$n$ sets $X, Y \subseteq ([1, \Phi]^d, \ell_1)$ with $\Phi = \poly(nd\varepsilon^{-1})$.
\Cref{obs:correctness of mwu} guarantees that: if the MWU algorithm in \Cref{fig:mwu} returns $(\alpha, \beta) \in \Gamma_t$, then $\emd(X, Y) \geq t$ by \Cref{cl:cert};
else, if MWU returns ``Fail'', then 
$\emd(X, Y) \leq (1+3\varepsilon)t$.

Moreover, \Cref{thm:main algorithm works} in \Cref{sec:sampling via cp} ensures that the MWU algorithm in \Cref{fig:mwu} can be implemented in time $\tilde O(n^{2-\Omega(\phi)} + nd)$ and with success probability $1- 1 / \poly(n, \log\Phi)$.
Thus, we can run an exponential search for $\emd(X, Y)$ with parameter $t_k = (1 +\varepsilon)^k$ and return the smallest $t_k$ such that the MWU algorithm in \Cref{fig:mwu} returns ``Fail''.
Then $\polylog(nd\eps^{-1})$ steps of binary search suffice.
\end{proof}
\section{Sampling from the Distribution $\lambda$ via Closest Pair}
\label{sec:sampling via cp}

The main result of this section is the following.

\begin{theorem}[Subquadratic implementation of MWU] \label{thm:main algorithm works}
Fix $\varepsilon, \phi > 0$.
Suppose that there exists a randomized algorithm for $(1+\varepsilon)$-approximate closest pair on $(\mathbb{R}^d, \ell_1)$ with running time $T_{CP}(n, \varepsilon) = n^{2-\phi}$ and success probability at least $2/3$.
Then, it is possible to implement the multiplicative weights update algorithm in \Cref{fig:mwu}
 
in time $\tilde O(n^{2-\Omega(\phi)} + nd)$ 
with success probability $1 - 1 /\poly(n, \log \Phi)$.
\end{theorem}
\begin{proof}

Notice that we can boost the success probability of the CP algorithm to $1 - 1 / \poly(n)$ by re-running it $O(\log n)$ times and returning the closest pair among those runs. Since we are going invoke the CP algorithm at most $\poly(n)$ times, we can condition on the event that the CP algorithm never fails.

Recall that $\sfD_\ell = \eps / \log \Phi$, $\sfD_u = O(\log n)$ and $h = O(\log \Phi)$, as defined at the beginning of \Cref{sec:template}. 
According to \Cref{lem:sep}, we can choose
\[ 
\gammagap = \frac{\eps \cdot \sfD_\ell}{ 2 h} \qquad \text{and}\qquad \sfK = \sfD_u \cdot \sfD_T, \qquad \text{ where } \qquad \sfD_T = O(\eps^{-1} \cdot \log n \cdot \log \Phi). 
\]
Likewise, according \Cref{lem:mwu}, we can choose 
\begin{align*}
\eta = \frac{\gammagap}{100 \cdot \sfK^2} \qquad \sfR = \frac{\ln(2n^2)}{100 \cdot \eta \gammagap} \qquad \text{and}\qquad \chi = \frac{\gammagap}{100 \cdot \sfR \sfK}.
\end{align*} 
Thus, $\sfK, \sfR, \eta^{-1}, \chi^{-1}, \gammagap^{-1} \leq \poly(\log(n \Phi), \varepsilon^{-1})$.
Finally, define $\tau := n^{1+\phi /2}$.

\paragraph{Algorithm.}
Upfront, we sample a set $S \subseteq X \times Y$ of size $n^{2-\phi/8}$ uniformly at random.
By \Cref{lem:mwu}, to implement the MWU algorithm in \Cref{fig:mwu}, we just need to correctly implement $\textsc{Certify}(\lambda^{r},\gammagap,\sfK, \delta)$ for $r = 1 \dots \sfR$.
By \Cref{lem:sep}, we can implement $\textsc{Certify}(\lambda^r,\gammagap,\sfK, \delta)$ as long as we can sample from a distribution $\lambda'$ such that $\|\lambda' - \lambda^r\|_{\tv} = o\left( \frac{\varepsilon \delta \sfD_\ell}{h \sfD_u}\right) = o\left(\frac{\varepsilon^2}{\log^2 n \cdot \log^2 \Phi}\right)$. In what follows, we set $\delta < 1/(100R) = 1/\polylog(n)$ so that we can guarantee correctness of the \textsc{Certify} procedure over all $R$ rounds (recall that $\delta$ is the failure probability of the \textsc{Certify} procedure). 
Recall that $\textsc{Certify}(\lambda^r,\gammagap,\sfK, \delta)$ returns $(\balpha, \bbeta)$ satisfying
\[
|\balpha_i - \bbeta_j| \leq C_{ij} \cdot \sfK \leq \poly(n, \Phi),
\]
hence throughout the execution of MWU we have $||\balpha||_\infty, ||\balpha||_\infty \leq \poly(n, \Phi)$.
Recall that $\lambda^r = \lambda^r(\eta, C, D^r, P^r)$ belongs to the class of distributions $\calD(\eta)$ defined in \Cref{def:disc-dist}. In the following, we show how to sample from $\lambda^r$ by leveraging Lemmas \ref{lem:from CP to all close pair retrieval}, \ref{lem:from all close pairs retrieval to sampling from lambda} and \ref{lem:from constant D to arbitrary D}.

\begin{definition}\label{def:level-set}
Fix any two sets $X,Y \subset \R^d$, such that $\|x-y\|_1 \in [1,\Phi]$ for all $x \in X, y \in Y$. Then for any $t \geq 0$, define the level set.
          \[
      L_t(X,Y) = \{(x, y) \in X \times Y \,|\, (1+\varepsilon)^{t-1} \leq ||a - b|| < (1+\varepsilon)^t\}
      \]
      for $t \in [\eps^{-1} \cdot \log \Phi]$.
\end{definition}

\begin{definition} \label{def:shattering function}
Given a partition $A = A_1 \cup \dots \cup A_k$, we say that $D \subseteq A$ $\tau$-\emph{shatters} the partition, for some $\tau \geq 1$, if for all $i$ such that $|A_i| \geq \tau$, we have
\[
0.9 \cdot \frac{|A_i|}{|A|} \le \frac{|D \cap A_i|}{|D|} \le 1.1 \cdot \frac{|A_i|}{|A|}.
\]
\end{definition}

Combining \Cref{lem:from CP to all close pair retrieval} and \Cref{lem:from all close pairs retrieval to sampling from lambda} we obtain the following algorithm.

\begin{framed}
\noindent $\constsampler$
\begin{itemize}
    \item \textbf{Input}: Sets $X', Y' \subseteq (\bbR^d, \ell_1)$ of size $m$, 
    a set $S \subseteq X' \times Y'$ of size $m^{2- \rho \phi}$ for some constant $\rho \in (0, 1/2)$, a matrix $P \in \{\pm 1\}^{m\times m}$, and a matrix $D \in \bbR^{m\times m}$ with all equal entries.
    \item \textbf{Output}: $n$ samples from $\lambda = \lambda(\eta, C(S), D, P)$, where $C(S)$ is defined as in \Cref{def:rounded-distances}.
\end{itemize} 
$\constsampler$ is correct with high probability, as long as $S$ $\tau$-shatters the collection $\{L_t(X', Y')\}_{t \in [\eps^{-1} \cdot \log \Phi]}$.
     $\constsampler$ runs in time $\tilde O(m^{2-\Omega(\phi)})$.
\end{framed}
Notice that the matrices $C, D, P$ fed as input to $\constsampler$ are represented implicitly according to Definitions \ref{def:rounded-distances} and \ref{def:rounded-duals}, allowing for subquadratic running time. 

Then, we invoke \Cref{lem:from constant D to arbitrary D}, which yields the following algorithm.

\begin{framed}
\noindent $\arbitsample$
\begin{itemize}
    \item \textbf{Input}: Sets $X, Y \subseteq (\bbR^d, \ell_1)$ of size $n$, a set $S \subseteq X \times Y$ of size $n^{2-\phi / 8}$, 
    dual solutions $(\balpha, \bbeta) \in \bbZ^{2n}$,
    a matrix $P \in \{\pm 1\}^{n\times n}$, and a matrix $D \in \bbR^{n\times n}$ defined as in \Cref{def:rounded-duals}. 
    
    \item \textbf{Output}: $n \cdot \log^z(n)$ samples from $\lambda'$ such that $|\lambda' - \lambda|_{TV} = o\left(\frac{\varepsilon^2 \delta^2}{\log^2 n \cdot \log^2 \Phi}\right)$, where $\lambda = \lambda(\eta, C(S), D, P)$, $C(S)$ is defined as in \Cref{def:rounded-distances},
    and $z > 0$ is any constant. 
\end{itemize} 
$\arbitsample$ calls 
\[
\constsampler(X_i, Y_i, S \cap (X_i \times Y_i), P|_{X_i\times Y_i}, D|_{X_i\times Y_i})
\]
for $i \in [k]$, where $\{X_i \times Y_i\}_{i \in  [k]}$ is a collection of disjoint combinatorial rectangles in $X \times Y$ of size $\sqrt[4]{n} \times \sqrt[4]{n}$ such that $D|_{X_i \times Y_j}$ has all equal entries.
$\arbitsample$ succeeds with high probability, as long as $\constsampler$ is correct on all these calls. 
Finally, $\arbitsample$ runs in time $\tilde O(n^{2-\Omega(\phi)})$.
\end{framed}
Denote with $\{X^r_i \times Y^r_i\}_{i\in [k]}$ the combinatorial rectangles defined by $\arbitsample$ at round $r$.
We would like to prove that
$S$ $\tau$-shatters the collection 
\begin{equation} \label{eq:collection to shatter}
\set{L_t(X^r_i, Y^r_i) \,:\, r\in [\sfR]\,, i \in [k] \text{ and } t = 1 \dots \varepsilon^{-1}\cdot \log \Phi)}.
\end{equation}
If that is the case, then: (i) the set $S \cap (X^r_i \times Y^r_i)$ fed to $\constsampler$ has the correct size; (ii) $\constsampler$ is correct, with high probability.
To this end, we invoke \Cref{cor:consistency lemma}.

\paragraph{Tracking dependencies.}
In order to verify the hypotheses of \Cref{cor:consistency lemma}, we need to track how the objects defined in our algorithm depend on the random choice of $S$.
When talking about dependencies, we consider all random variables besides $S$ fixed.

Using the notation of \Cref{cor:consistency lemma}, we have $A = [n] \times [n]$ and $p^r = \lambda^r$.
The state $\sigma_r$ corresponds to the variables $(\balpha^r, \bbeta^r)$, which, in turn, define the matrices $P^r$ and $D^r$.
The partition function 
\[
f_r: [n]\times [n] \to \set{L_t(X^r_i, Y^r_i) \,:\, i \in [k] \text{ and } t = 1 \dots \varepsilon^{-1}\cdot \log \Phi)} \cup \{\bot\}
\]

corresponds to the assignment of pairs to combinatorial rectangles (or, if $f(i, j) = \bot$, the fact that the pair $(i,j)$ is not assigned to any rectangle).
Notice that $f_r$ depends on the state $\sigma^r$ (through $D^r$).
 The state $\sigma_{r+1} = (\balpha^{r+1}, \bbeta^{r+1})$ depends on $\sigma_r = (\balpha^{r}, \bbeta^{r})$ as well as the the $\tilde O(n)$ samples returned by $\arbitsample$ at the previous step.
Since the value of $\lambda^r(i, j, \sigma)$ only depends on $P^r_{ij}$ and $D^r_{ij}$, the distribution $\lambda^r$ only depends on $\sigma^r$.

\paragraph{Analysis.}
Above, we verified that the collection in \Cref{eq:collection to shatter} is shattered by the set $S$, which implies that $\constsampler$ is correct by \Cref{lem:from all close pairs retrieval to sampling from lambda}.
Therefore, the correctness of $\arbitsample$ is guaranteed by \Cref{lem:from constant D to arbitrary D}.
By running \textsc{Certify} with a sufficiently small parameter $\delta = 1 / \polylog(n\Phi)$ we ensure correctness across all $\sfR$ rounds.
If \textsc{Certify} is correct, then by \Cref{obs:correctness of mwu}, the MWU algorithm is also correct.

The running time of the algorithm is $\tilde O(n^{2-\Omega(\phi)} + nd)$.
In fact, the algorithm runs $\sfR \leq \poly(\log(n\Phi), \varepsilon^{-1})$ rounds, and in each round it solves \textsc{Certify} using $\tilde O(n)$ samples and $O(nd \Phi)$ time (\Cref{lem:sep}). Finally, to compute the $\tilde O(n)$ samples, the algorithm $\arbitsample$ uses time $\tilde O(n^{2-\phi})$.

\end{proof}

\subsection{From Closest Pair to All Close Pairs Retrieval} \label{sec:from CP to all close pair retrieval}

Recall that $L_t(X,Y)$ is define as in \Cref{def:level-set}.
When the sets $X,Y$ are fixed from context, we will write $L_t = L_t(X,Y)$. Further, we define the \emph{prefix set} $\mathcal L_t = \mathcal L_t(X,Y) = \bigcup_{j \leq t} L_j(X,Y)$.

\begin{lemma} \label{lem:from CP to all close pair retrieval}
      Suppose that there exists an algorithm that solves $(1+\varepsilon)$-approximate closest pair in time $T(n, \varepsilon) = n^{2-\phi}$.
      Let $X, Y \subseteq ( \bbR^n, \ell_1)$ be size-$n$ sets such that $\|x - y\| \in [1, \Phi]$ for each $(x, y) \in X\times Y$. 
      Then, the algorithm $\findclosepairs$ in \Cref{fig:findclosepairs} returns and integer $t \geq 0$ along with the prefix set $\mathcal L_t$ such that $|\mathcal L_t| = \tilde O(n^{1 + \phi})$ and the level set $L_{t+3}$ satisfies $|L_{t+3}| = \tilde \Omega( n^{1+\phi})$. $\findclosepairs$ runs in time 
      $\tilde O(n^{2 - \Omega(\phi)})$.
\end{lemma}

In the remainder of this section, we prove \Cref{lem:from CP to all close pair retrieval}.

\begin{figure}[h]
\begin{framed}
\noindent \textbf{Algorithm} $\sscp(X, Y, \eta)$. 
\begin{enumerate}
\item Let $X'$ (resp. $Y'$) be the set obtained by subsampling $X$ (resp. $Y$) with rate $\eta$.
\item Return a $(1+\varepsilon)$-approximate closest pair in $X' \times Y'$, with success probability $1- n^{-3}$. \label{fig: closest pair in sscp}
\end{enumerate}
\end{framed}
\caption{Implementation of $\sscp$. \label{fig:sscp}}
\end{figure}

\begin{lemma} \label{lem:sscp running time}
If $|X| = |Y| = n$, $\sscp$ takes time $ \tilde O( (\eta \cdot n)^{2-\Omega(\phi)})$, with high probability.
\end{lemma}
\begin{proof}
First, observe that subsampling with rate $\eta$ can be implemented in time $\tilde O(\eta \cdot n) = \tilde O( (\eta \cdot n)^{2-\Omega(\phi)})$.
In fact, it is sufficient to sample a binomial random variable $b \sim \text{Bin}(n, \eta)$ and then sample $b$ distinct elements from $X$ (resp. $Y$).
By standard concentration bounds $b = \tilde O(\eta \cdot n)$ with high probability, hence the claimed running time for the subsampling step.

By conditioning on the event $|X'|, |Y'| = \tilde O(\eta \cdot n)$, computing a $(1+\varepsilon)$-approximate closest pair takes time $ \tilde O( (\eta \cdot n)^{2-\phi})$. Indeed, if $|X'| \neq |Y'|$ we can simply pad the smallest set with dummy (far-from-all-points) vectors.
\end{proof}

\begin{figure}[h]
\begin{framed}
\noindent \textbf{Algorithm} $\llp(X, Y, z)$.
\begin{enumerate}
\item Let $S \subseteq X \times Y$ be a set of $\frac{n^2}{z^2} \log \Phi$ i.i.d. uniform samples from $X \times Y$.
\item Let $t + 3:= \min \set{s\geq 0 \,|\, L_s \cap S \neq \emptyset}$.
\item Return $t$.
\end{enumerate}
\end{framed}
\caption{Implementation of $\llp$. \label{fig:llp}}
\end{figure}

\begin{lemma} \label{lem:t is correct}
Fix $z \in [\sqrt{n}, n]$. Let $t = \llp(X, Y, z)$, then $|L_{t+3}| \geq \frac{z^2}{\log^3 \Phi}$ and $|\cL_{t+2}| \leq 0.1 \cdot z^2 $, with probability $1 -o(1)$.
\end{lemma}
\begin{proof}
For any $s = 1 \dots \varepsilon^{-1} \cdot \log \Phi$, if $|L_s| < \frac{z^2}{\log^3 \Phi}$ we have
\begin{align*}
\Pr[L_s \cap S \neq \emptyset] &\leq \\
\sum_{(x, y) \in L_s} \Pr[(x, y) \in S] &< \\
\frac{z^2}{\log^3 \Phi} \cdot \left(1 - \left(1 - \frac{1}{n^2}\right)^{\frac{n^2}{z^2} \log \Phi} \right) &= \\
\frac{z^2}{\log^2 \Phi} \cdot O\left(\frac{\log \Phi}{z^2}\right) &= 
O\left(\frac{1}{\log^2 \Phi}\right).
\end{align*}
Thus, by taking a union bound over all $\varepsilon^{-1} \cdot \log \Phi$ possible values of $s$ we have that, with probability $1-o(1)$, $|L_s| < \frac{z^3}{\log^3 \Phi}$ implies $L_s \cap S = \emptyset$. Thus, $|L_{t+3}| \geq \frac{z^2}{\log^3 \Phi}$.

Let $s$ be the smallest integer in $[\eps^{-1} \cdot \log \Phi]$ such that $\cL_s > 0.1 \cdot z^2$.
Then, 
\begin{align*}
\Pr[\cL_s \cap S = \emptyset] &\leq
\left(1 - \frac{0.1 \cdot z^2}{n^2}\right)^{\frac{n^2}{z^2} \log \Phi} = \frac{1}{\Phi^{\Omega(1)}}. 
\end{align*}
Thus, with probability\footnote{
If $\Phi = \omega(1)$ does not hold, we can replace $\Phi$ with $\max\{\Phi, n\}$.
} $1- o(1)$, we must have $t+3 \leq s$, which implies $\cL_{t+2} \leq 0.1 \cdot z^2$.

\end{proof}

\begin{figure}[H]
\begin{framed}
\noindent \textbf{Algorithm} $\findclosepairs(X, Y)$. 
\begin{enumerate}
\item Let $z = \sqrt{n^{1+\phi}}$.
\item $t = \llp(X, Y, z)$.
\item Initialize the set $L = \emptyset$.
\item For $i = 1 \dots \tilde O(z^2)$: 
\begin{enumerate}[(a)]
    \item Let $(x, y) = \sscp(X, Y, 1/z)$. \label{fig:collect light edges}
    \item If $(x, y) \in \cL_t$, add $(x, y)$ to $L$.
\end{enumerate}
\item \textbf{Invariant (I):} $L$ contains all light edges in $\cL_t$.
\item Initialize the counter $c: X \cup Y \rightarrow \bbZ$ to all zeros.
\item For $i = 1 \dots T = \tilde O(z)$:
\begin{enumerate}[(a)]
    \item Let $(x, y) = \sscp(X, Y, 1/z)$.  
    \item If $(x, y) \in \cL_{t+1}$, set $c[a] \leftarrow c[a] + 1$ and $c[b] \leftarrow c[b] + 1$.
\end{enumerate}
\item Define $F$ as the set of $x \in X \cup Y$ with $c(x) \geq 0.02 \cdot T$. \label{fig:if c of x}
\item \textbf{Invariant (II):} $\set{0.03\text{-frequent vertices}} \subseteq F \subseteq \set{0.01\text{-frequent vertices}}$.
\item Initialize the set $H = \emptyset$.
\item For $f \in F$: \label{fig:for loop brute force search}
\begin{enumerate}[(a)]
    \item If $f \in X$, then add to $H$ all $(f, b) \in \{f\} \times Y$ such that $(f, b) \in \cL_t$.
    \item If $f \in Y$, then add to $H$ all $(a, f) \in X \times \{f\}$ such that $(a, f) \in \cL_t$.
\end{enumerate}
\item \textbf{Invariant (III):} $H$ contains all heavy edges in $\cL_t$.
\item Return $t, L \cup H$.
\end{enumerate}
\end{framed}
\caption{Implementation of $\findclosepairs$. \label{fig:findclosepairs}}
\end{figure}

Throughout, fix $z = \sqrt{n^{1+\phi}}$ and $t$ as in $\findclosepairs$ and condition on the high-probability event in \Cref{lem:t is correct}.
Consider the bipartite graph $G=(X \cup Y, \cL_{t+1})$. We say that a vertex $x \in X \cup Y$ is \emph{light} if $\deg_{\cL_{t+1}}(x) \leq 0.5 \cdot z$, and \emph{heavy} otherwise.
We say that an edge $(x, y) \in \cL_{t+1}$ is light if both its endpoints are light, and heavy otherwise. 
We say that a vertex $x \in X\cup Y$ is $C$-\emph{frequent} if 
\[
\Pr\left[\sscp(X, Y, 1/z) \text{ returns and edge from } \cL_{t+1} \text{ incident to } x \right] \geq C \cdot z^{-1}.
\]

\begin{lemma} \label{lem:collecting light edges is likely}
Let $(x, y) \in \cL_t$ be a light edge. Then, $\Pr[\sscp(X, Y, 1/z) \text{ returns } (x, y)] = \Omega(z^{-2})$.
\end{lemma}
\begin{proof}
The probability that both $x$ and $y$ are subsampled (namely $(x, y) \in X' \times Y'$) is exactly $z^{-2}$.
Let $\cF$ be this event.
Let $\cE$ be the event that some other edge $(x, y) \in \cL_{t+1}$ belongs to $X' \times Y'$. Since $(x, y) \in \cL_t$, if $\cF$ happens and $\cE$ does not happen, then our closest-pair algorithm correctly returns $(x, y)$. Our goal is now to upper bound the probability of $\cE$ conditioned on $\cF$.

Let $\cE_1$ be the event that there exists an edge $(a, b)$ with $a \neq x$ and $b \neq y$ such that $(a, b) \in X' \times Y'$.
 Let $\cE_2$ be the event that there exists an edge of the form $(a, y)$ (with $a \neq x$) or $(x, b)$ (with $b \neq y$) that belongs to $X'\times Y'$.
Clearly, $\cE = \cE_1 \cup \cE_2$, so, it is sufficient to show that $\Pr[\cE_1 \,|\, \cF] + \Pr[\cE_2 \,|\, \cF] \leq 1 - \Omega(1)$.

The event $\cE_1$ is independent of $\cF$, so we can bound
\[
\Pr\left[ \cE_1 \,|\, \cF \right] = \Pr [\cE_1] \leq |\cL_{t+1}| \cdot z^{-2} \leq 0.1,
\]
where the last inequality follows from $\Cref{lem:t is correct}$.

Since $(x, y)$ is light, then the $\cL_{t+1}$-neighborhoods of $x$ and $y$ are small. More precisely, we have $|N_{\cL_{t+1}}(b)|, \allowdisplaybreaks|N_{\cL_{t+1}}(a)| \leq 0.5 \cdot z$. So,
\[
\Pr\left[N_{\cL_{t+1}}(a) \cap Y' \neq \emptyset \right] = 1 - \left(1-\frac 1 z \right)^{|N_{\cL_{t+1}}(a)|} \leq 1 - \left(1-\frac 1 z \right)^{0.5 \cdot  z} = 1 - \frac{1}{e^{0.5}} + o(1) < 0.4.
\]
and likewise for $\Pr\left[N_{\cL_{t+1}}(b) \cap X' \neq \emptyset \right]$.
Therefore,
\[
\Pr\left[ \cE_2 \,|\, \cF \right] \leq \Pr\left[N_{\cL_{t+1}}(a) \cap Y' \neq \emptyset \right] + \Pr\left[N_{\cL_{t+1}}(b) \cap X' \neq \emptyset \right] \leq 0.8.
\]
Finally, $\Pr[\cE \cond \cF] \leq 0.9$, thus
\[
\Pr[\cF \text{ and } \neg\cE] \geq \Pr[\cF] \cdot \Pr[\neg \cE \,|\, \cF] \geq 0.1 \cdot z^{-2}.
\]
\end{proof}

\begin{lemma} \label{lem:all heavy are frequent}
Every heavy vertex is $0.03$-frequent.
\end{lemma}
\begin{proof}
Consider a heavy vertex $x \in X$ (the case $x \in Y$ is symmetric). Let $\cE$ be the event ``$\sscp(X, Y, 1/z) \text{ returns and edge from } \cL_{t+1} \text{ incident to } x$''.
We need to prove that $\Pr[\cE] \geq 0.03 \cdot z^{-1}$.
Let $\cF$ be the event $x \in X'$. 
Let $\cG$ be the event $\cL_{t+2} \cap (X' \setminus\{x\} \times Y') = \emptyset$
Denote with $N_{\cL_{t+1}}(x)$ is the neighborhood of $x$ with respect to edges in $\cL_{t+1}$.
Let $\cH$ be the event $N_{\cL_{t+1}}(x)  \cap Y' \neq \emptyset$.
Notice that $\cF$ and $\cG \cap \cH$
are independent.
Moreover, $\cE \supseteq \cF \cap \cG \cap \cH$.
So, $\Pr[\cE] \geq \Pr[\cF] \cdot \Pr[\cH\cap\cG]$.

By definition of subsampling we have $\Pr[\cF] = z^{-1}$. 
By union bound over all $|\cL_{t+2}| \leq 0.1 \cdot  z^2$
we have
\[
\Pr[\neg \cG] \leq  \sum_{(x, y) \in \cL_{t+2}} \frac 1 z \leq 0.1.
\]
Since $x$ is heavy we have $|N_{\cL_{t+1}}(x)| > 0.5 \cdot z$, so
\[
1 - \Pr[\cH] = \Pr[\neg \cH] = \left(1 - \frac{1}{z}\right)^{|N_{\cL_{t+1}}(x)|} < \left(1 - \frac{1}{z}\right)^{0.5 \cdot z} < 0.65.
\]
Thus, $\Pr[\cG \cap \cH] \geq \Pr[\cH] - \Pr[\neg \cG] \geq 0.35 - 0.1 = 0.25$ and we have
\[
\Pr[\cE] \geq z^{-1} \cdot 0.25 \geq 0.03 \cdot z^{-1}.
\]
\end{proof}

\begin{lemma} \label{lem:frequent vertices are few}
There are at most $O(z)$ $0.1$-frequent vertices.
\end{lemma}
\begin{proof}
Let $x \in X$ (the case $x \in Y$ is symmetric).
First, we observe that if $\deg_{\cL_{t+1}}(x) < 0.01 \cdot z$, then $x$ is not $0.01$-frequent.
Indeed,
\begin{align*}
\Pr\left[\sscp(X, Y, 1/z) \text{ returns and edge from } \cL_{t+1} \text{ incident to } x \right] &\leq \\
\Pr[x \in X'] \cdot \Pr \left[  Y' \cap N_{\cL_{t+1}}(x) \neq \emptyset \right] =
\Pr[x \in X'] \cdot \left(1 - \left(1-\frac 1 z\right)^{|N_{\cL_{t+1}}(x)|} \right) &\leq \\
\Pr[x \in X'] \cdot \left(1 - \left(1-\frac 1 z\right)^{0.01 \cdot z} \right) &< 0.01 \cdot z^{-1}.
\end{align*}

By \Cref{lem:t is correct}, we have $|\cL_{t+1}| \leq |\cL_{t+2}| =O(z^2)$. Thus, a counting argument shows that there are at most $O(z)$ many $x \in X \cup Y$ with $\deg_{\cL_{t+1}}(x) \geq 0.01 \cdot z$.
\end{proof}

In the end, we prove \Cref{lem:from CP to all close pair retrieval}.
\begin{proof}[Proof of \Cref{lem:from CP to all close pair retrieval}]
Thanks to \Cref{lem:t is correct}, we have that $|\cL_{t}| = \tilde O (z^2) = \tilde O(n^{1+\phi})$ and $|L_{t+3}| = \tilde \Omega(z^2) = \tilde \Omega(n^{1+\phi})$.
Now, using the notation of \Cref{fig:findclosepairs}, we need to prove invariants (I) and (III), namely, that $L$ (resp. $H$) contains all the light (resp. heavy) edges in $\cL_t$.

Let $(x, y) \in \cL_t$ be a light edge.
By \Cref{lem:collecting light edges is likely}, the probability of collecting $(x, y)$ at line \ref{fig:collect light edges} is $\Omega(z^{-2})$, thus $\tilde O(z^2)$ iterations are sufficient to collect $(x, y)$ with high probability. Then, a union bound over all light edges shows that, with high probability, $L$ contains all light edges in $\cL_t$.

Since the $T = \tilde O(z^2)$ executions of line \ref{fig:collect light edges} are independent, standard concentration bounds ensure that for each $0.03$-frequent vertex $x$, $c(x) \geq0.02 \cdot T$ at line \ref{fig:if c of x}, and thus $x \in F$, with high probability. Likewise, with high probability, each vertex $x$ that is not $0.01$-frequent satisfies $c(x) < 0.02 \cdot T$ at line \ref{fig:if c of x} and thus $x \not \in F$. This proves invariant (II).
Condition on invariant (II).
By \Cref{lem:all heavy are frequent}, we have that all heavy vertices are $0.03$-frequent, and thus all heavy edges are incident to $F$. Since $\findclosepairs$ runs an exhaustive search over the the sets $(X \cap F) \times Y$ and $X \times (Y \cap F)$, then invariant (III) must hold.

Finally, we bound the running time of $\findclosepairs$.
We make $\tilde O(z^2)$ calls to $\sscp(X, Y, 1/z)$.
Thanks to \Cref{lem:sscp running time}, each of these calls takes time $\tilde O((n/z)^{2-\Omega(\phi)})$ with high probability. Thus, the combined time spent running $\sscp$ is 
\[
\tilde O(n^{1+\phi}) \cdot \left(n^{\frac{1}{2} - \frac \phi 2}\right)^{2-\Omega(\phi)}= \tilde O(n^{1+\phi}) \cdot n^{1 - (1 + \Omega(1)) \phi + \Theta(\phi^2)} = \tilde O(n^{2-\Omega(\phi)}).
\]
The for loop at line \ref{fig:for loop brute force search} takes time $O(|F| \cdot n)$. Combining \Cref{lem:frequent vertices are few} and invariant (II), we obtain $|F| = O(z)$, thus the for loop at line \ref{fig:for loop brute force search} takes $O(n^{3/2 + \phi / 2}) = \tilde O(n^{2 - \Omega(\phi)})$ time.

\end{proof}

\subsection{From All Close Pairs Retrieval to Sampling from $\lambda$, for Fixed $D$} \label{sec:from all close pairs retrieval to sampling from lambda}

The following theorem, is going to be employed in the remainder of this section.

\begin{theorem}[Theorem 4 in \cite{beretta2024better}] \label{lem:estimating normalization constant}
Let $\mu$ be a distribution supported over a size-$n$ set $X$. Let $w : X \rightarrow \bbR_+$ be a function such that $\Pr_{\bm y \sim \mu}[\bm y = x] \propto w(x)$.
Suppose that we can sample $\bm y \sim \mu$ and observe $w(\bm y)$.
Then, there exists an algorithm that takes $O(\sqrt{n} / \varepsilon)$ samples and estimates $W:= \sum_{x\in X} w(x)$ up to a multiplicative error $1\pm \varepsilon$. 
\end{theorem}

\begin{lemma}[Sampling from $\lambda$ for Fixed $D$]
\label{lem:from all close pairs retrieval to sampling from lambda}
Fix any $\phi, \eta \in (0,1/2)$, $n \geq 1$, and $X, Y \subseteq (\bbR^d, \ell_1)$.
Let $S \subseteq X \times Y$ be a set of size $n^{2- \rho \cdot \phi}$ for some constant $\rho \in (0, 1/2)$.
Let $C = C(S) \in \bbR^{n \times n}$ be a rounded distance matrix defined as in \Cref{def:rounded-distances}.
Fix an arbitrary matrix $P \in \{\pm 1\}^{n\times n}$ and let $D \in \bbR^{n \times n}$ be a matrix of all equal entries. 
Define the distribution $\lambda = \lambda(\eta, C, D, P)$ as in \Cref{def:disc-dist}. 

Suppose that we are given $L_t(X,Y)$ such that $|\mathcal L_t(X,Y)| =  \tilde O(n^{1 + \phi})$ and $|L_{t+3}(X,Y)| = \tilde \Omega (n^{1+\phi})$.
Furthermore, suppose that we are given $S$, and that $S$ $\tau$-shatters the partition $\{L_\ell\}_{\ell\geq 0}$ of $X \times Y$ (\Cref{def:shattering function}), for $\tau = n^{1+\phi/2}$.
Then, there exists an algorithm $\constsampler$ that, with probability at least $1-1/\poly(n)$,
generates $n$ samples from $\lambda$, in time $\tilde{O}(n^{2-\Omega(\phi)})$.
\end{lemma}

\begin{proof}
Let $\kappa = D_{i,j}$ be the fixed value of all entries of $D$. Then recall that $\lambda_{i,j,\sigma} \propto w_{i,j,\sigma}$, where 
\[ w_{i,j,\sigma} =  \exp\left (\eta \cdot \kappa \cdot \sigma\cdot\frac{   P_{ij} }{C_{ij}} \right). \]

Where recall that $C_{i,j} = (1+\eps)^{\psi_{i,j} - 2S_{i,j}}$ where $\psi_{i,j}$ is the integer satisfying $(1+\eps)^{\psi_{i,j}} \leq \|x_i - y_j\|_1 \leq (1+\eps)^{\psi_{i,j}+1}$.
To sample $(i,j,\sigma)$ with probability proportional to $w_{i,j,\sigma}$, it suffices to first sample $(i,j)$ with probability proportional to 

\[ w_{i,j} =  \exp\left (\eta \cdot |\kappa| \cdot\frac{ 1}{C_{ij}} \right). \]

And then sample $\sigma \in \{1,-1\}$ with probability exactly $\frac{w_{i,j,\sigma}}{2 \cdot w_{i,j}}$, and reject if neither $\sigma = 1$ or $\sigma = -1$ is sampled. If sampled, it is clear that the resulting $(i,j,\sigma)$ is drawn from $\lambda$. Finally, note that for exactly one value of $\sigma \in \{-1,1\}$ we have $w_{i,j} = w_{i,j,\sigma}$, and therefore the probability that a sample is rejected is at most $1/2$.
In what follows, let $\lambda_{i,j}$ be the distribution where $\pr{(i,j)} \propto w_{i,j}$.

To sample proportionally to $w_{i,j}$, first, the algorithm exactly computes the values $w_{i,j}$ for all $(i,j) \in \mathcal{T} := \calL_t \cup \left(S \cap \left(L_{t+1} \cup L_{t+2} \right) \right)$, and the normalization factor $w_\calT = \sum_{(i,j) \in \calT} w_{i,j} $. For a set $T\subset [n] \times [n]$, let $\overline{T}$ be its complement, and let  $\lambda|_{T}$ denote the distribution of $\lambda$ conditioned on $T$. We show how to sample from both $\lambda|_{\calT}$ and $\lambda|_{\overline{\calT}}$. 

First, for any $(i,j) \in \overline{\calT}$, we know that $C_{i,j} \geq (1+\eps)^t$, since (1) we know $(i,j) \in L_{t'}$ for some $t' > t$ and (2) $(i,j) \notin S$ unless $t' \geq t+3$ by definition. Set $w_{\max} = \exp(\eta \cdot |\kappa| \cdot (1+\eps)^{-t})$. Then $w_{i,j} \leq w_{\max}$ for all $(i,j) \in \overline{\calT}$. Thus, to sample from $\lambda|_{\overline{\calT}}$, we can perform the following steps: \textbf{(1)} sample $(i,j) \sim [n] \times [n]$ uniformly, \textbf{(2)} reject $(i,j)$ if $(i,j) \notin \overline{\calT}$, \textbf{(3)} keep the sample with probability $\frac{w_{i,j}}{w_{\max}}$, otherwise reject the sample. Clearly the distribution is correct, thus it remains to analyze the failure probability. First note that $|\overline{\calT}| \geq n^2 - O(n^{2- \rho \cdot \phi} + n^{1+ \phi}) \geq n^2/2$, thus the probability of rejection in step \textbf{(2)} is at most $1/2$. Next, note that by definition of $\tau$-shattering, noting that $|L_{t+3}| \geq \tau$, we have that 
\[ \frac{|S \cap L_{t+3}|}{|S|} = \left(1 \pm \frac{1}{10}\right) \frac{|L_{t+3}|}{n^2} = \tilde{\Omega}(n^{-1 + \phi} )  \]
Thus $|S \cap L_{t+3}| \geq n^{1+\phi/2}$, and note that $w_{i,j} = w_{\max}$ for all $(i,j) \in S \cap L_{t+3}$. Thus, whenever $(i,j)$ is sampled from $S \cap L_{t+3}$, we do not reject in step \textbf{(3)} in the above sampling procedure. Moreover, we sample such an $(i,j)$ on step \textbf{(1)} with probability at least $\frac{|S \cap L_{t+3}|}{n^2} \geq \frac{1}{n^{1-\phi/2}}$. By Chernoff bounds (applied to the event that we successfully sample a $(i,j)$ or reject it), with probability $1-1/\poly(n)$ it follows that we can compute $n$ samples from $\lambda|_{\overline{\calT}}$ in time $\tilde O(n^{2-\Omega(\phi)})$, as needed.

Finally, note that after $O(|S| + |\cL_t|) = n^{2-\Omega(\phi)}$ pre-processing time to compute $w_{i,j}$ for all $(i,j) \in \calT$ as well as $w_\calT$, we can sample from $\lambda|_{\calT}$ in constant time. It remains to decide with what probability to sample from $\lambda|_{\calT}$ versus $\lambda|_{\overline{\calT}}$. Using Lemma \Cref{lem:estimating normalization constant}, we can compute a value $\hat{w_{\overline{\calT}}}$ such that $\hat{w_{\overline{\calT}}}  = (1 \pm 1/2)w_{\overline{\calT}} $ where  $w_{\overline{\calT}} = \sum_{(i,j) \in \overline{\calT}} w_{i,j}$. To sample from $\lambda$, we first chose to sample from $\calT$ or $\overline{\calT}$ with probability $\frac{w_{\calT}}{w_{\calT} + \hat{w_{\overline{\calT}}}}$ or $\frac{\hat{w_{\overline{\calT}}}}{w_{\calT} + \hat{w_{\overline{\calT}}}}$  respectively. Conditioned on the procedure not rejecting a sample, our procedure samples a pair $(i,j)$ with probability $\hat{\lambda_{i,j}}$ where $\hat{\lambda_{i,j}} = (1 \pm 1/2) \lambda_{i,j}$, since the only source of error is in our approximation $\hat{w_{\overline{\calT}}}$ of $w_{\overline{\calT}}$. 

We now show how to sample exactly from the distribution $\lambda_{i,j}$. To see this, note that whenever we sample a pair $(i,j)$ from $\hat{\lambda}$, we can compute exactly the probability $\hat{\lambda_{i,j}}$ with which we sampled it; this is clear in the case $(i,j) \in \calT$, and otherwise the probability is precisely $\frac{1}{|\overline{\calT}|} \cdot \frac{w_{i,j}}{w_{\max}}$, and $|\overline{\calT}| = n^2 - |\calT|$ can be computed exactly since $\calT$ is computed exactly.
Thus, to sample exactly from $\lambda$, we first sample $(i,j)$ from $\hat{\lambda}$, and reject it with probability $1-\frac{1}{\hat{\lambda_{i,j}}} \cdot \frac{w_{i,j}}{4(w_{\calT} + \widehat 
{w_{\overline{\calT}}})}$. First note that this is a valid probability, because $\frac{w_{i,j}}{4(w_{\calT} + \widehat 
{w_{\overline{\calT}}})} < \lambda_{i,j} / 2 \leq \hat{\lambda_{i,j}}$. Overall, the probability we sample and keep $(i,j)$ is
\[       \hat{\lambda_{i,j}} \cdot  \frac{1}{\hat{\lambda_{i,j}}} \cdot \frac{w_{i,j}}{4( w_{\calT} + \widehat 
{w_{\overline{\calT}}})}  =\frac{w_{i,j}}{4 (w_{\calT} + \widehat 
{w_{\overline{\calT}}})} \]
which is proportional to $w_{i,j}$. 
Thus conditioned on keeping the pair, we sample from precisely the correct distribution. Finally, note that $\frac{w_{i,j}}{4 (w_{\calT} + w_{\overline{\calT}})}> \frac{1}{8}\lambda_{i,j} > \frac{1}{16} \hat{\lambda_{i,j}}$, thus we reject with probability at most $15/16$, which can be corrected for by oversampling by a $O(1)$ factor, completing the proof.

\end{proof}

\subsection{From Fixed $D$ to Arbitrary $D$} \label{sec:from monotone sampling to lambda}

\begin{lemma} \label{lem:from constant D to arbitrary D}
Fix any $\eta, \eps, \phi \in (0, 1/2)$, $\chi \geq 1 / \poly(\varepsilon^{-1}, \log(n\cdot \Phi))$, 
$(\balpha, \bbeta) \in \bbZ^{2n}$ with $||\balpha||_\infty, ||\bbeta||_\infty \leq \poly(n  \Phi)$,
and, $C \in \bbR^{n \times n}$ be a distance matrix.
Define $P \in \{\pm 1\}^{n \times n}$ and $D \in \bbR^{n \times n}$ as in \Cref{def:rounded-duals}.
Define $\lambda \in \calD(\eta, C, D, P)$ as in \Cref{def:disc-dist}.

Suppose that there exists a randomized algorithm $\constsampler$ that returns $n$ samples from $\lambda$ and runs in time $\tilde O(n^{2-\Omega(\phi)})$, assuming that $D$ has all-equal entries.
Then, for each constant $z >0$, there exists an algorithm $\arbitsample$ that returns $n \cdot \log^z(n)$ samples from a distribution $\lambda'$, with 
$|\lambda' - \lambda|_\tv = o\left( \frac{\varepsilon^2}{\log n \cdot \log^2 \Phi}\right)$
and runs in time $\tilde O(n^{2-\Omega(\phi)} \cdot \poly(\eps^{-1}))$, for any $D$ defined in \Cref{def:rounded-duals}. 

Moreover, $\arbitsample$ calls $\constsampler$ on sets $\{X_i \times Y_i\}_{i \in [k]}$, were $k = O(n^{3/2})$, $|X_i| = |Y_i| = \sqrt[4]{n}$ and $\{X_i \times Y_i\}_{i \in [k]}$ are pairwise disjoint. If $\constsampler$ is correct on all such calls, then $\arbitsample$ is correct, with high probability.
\end{lemma}
\begin{proof}
Our proof strategy is as follows: we partition the triples $i,j,\sigma$ 
into groups where the entries of $D$ are constant, 
then first sample one of the groups and finally resort to
$\constsampler$ to sample a pair within the group.

\Cref{def:rounded-duals} and $||\balpha||_\infty, ||\bbeta||_\infty \leq \poly(n \Phi)$ imply that the entries of $D$ belong to $\{ 0 \} \cup \{ (1+\chi)^s :\, s \in \bbZ \cap [-T, T]\}$ for $T = O(\chi^{-1} \cdot \log (n \Phi))$. 
Define $Q_s := \{(i, j) \in [n]^2 \cond D_{ij} = (1+\chi)^s\}$, $Q_{\circ} := \{(i, j) \in [n]^2 \cond D_{ij} = 0\}$ and notice that $Q_{\circ} \cup \bigcup_{s \in [-T, T]} Q_s$ is a partition of $[n]^2$.
Our strategy is to construct a partition of $[n]\times [n]$ given by $E \cup \bigcup_{\ell \in [L]} I_\ell \times J_{\ell}$ satisfiying the following desiderata:
\begin{enumerate}[(i)]
    \item $|I_\ell| = |J_\ell| = \sqrt[4]{n}$'
    \item $I_\ell \times J_\ell \subseteq Q_s$ for some $s \in \{\circ\} \cup [-T, T]$.
    \item $|E| = O(n^{7/4})$.  
\end{enumerate}
We construct such partition processing one $Q_s$ at a time. Then, our sampling strategy will first sample a set $S \in \{E\} \cup \{I_\ell \times J_\ell\}_{\ell \in [L]}$ and then sample a triplet $(i, j, \sigma) \in S \times \{\pm 1\}$ via $\constsampler$. 

\paragraph{Computing our partition.}
Here we describe the algorithm that constructs the desired partition. Set $\ell \leftarrow 0$ and $E \leftarrow \emptyset$.
We start from $Q_\circ$.
For each $x \in \{\balpha_i\}_{i \in [n]} \cup \{\bbeta_i\}_{i \in [n]}$, let $A_x = \{i \in [n] \,|\, \balpha_i = x\}$ and $B_x = \{i \in [n] \,|\, \bbeta_i = x\}$.
It is straightforward to implicitly compute the decomposition 
\[
Q_\circ = \bigcup_{x \in \{\balpha_i\}_i \cup \{\bbeta_i\}_i} A_x \times B_x
\]
in near-linear time by sorting $\balpha$ and $\bbeta$.

For each $x \in \{\balpha_i\}_1 \cup \{\bbeta_i\}_1$, if $|A_x|, |B_x| \geq \sqrt{n}$, then define $I_\ell \leftarrow A_x$, $J_\ell \leftarrow B_x$, and update $\ell \leftarrow \ell + 1$. 
Else, update $E \leftarrow E \cup A_x \times B_x$.
      
For each $s \in [-T, T]$ we do the following.
Sort $\balpha_i$ so that $i \mapsto \balpha_i$ is increasing, and do the same for $\bbeta_j$.
For all $i \in [n]$, the set $Q_s \cap \{i\} \times [n]$ is an interval and we denote it with $[a(i), b(i)]$.
Let $q_j = j \cdot \sqrt{n}$ for $j \in [\sqrt
n]$. Then, we execute the loop in \Cref{fig:paritioning Qs}. 
      
\begin{figure}[H]
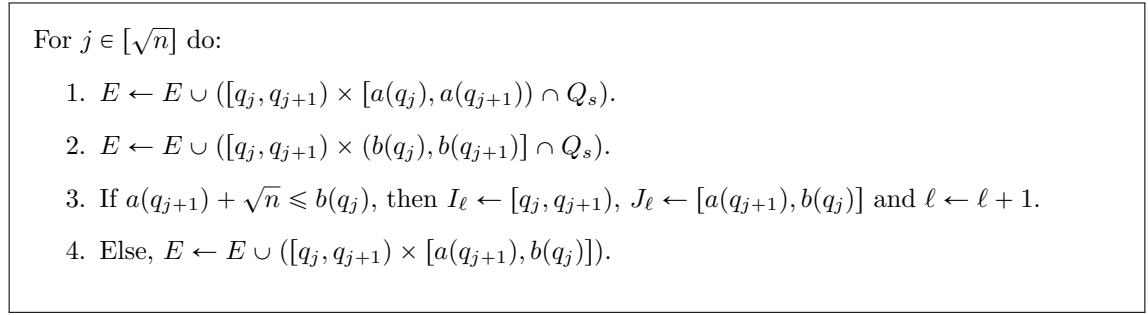

\begin{framed}
 For $j \in [\sqrt n]$ do:
      \begin{enumerate}
          \item $E \leftarrow E \cup \left([q_j, q_{j+1}) \times [a(q_j), a(q_{j+1})) \cap Q_s\right)$.
          \item $E \leftarrow E \cup \left([q_j, q_{j+1}) \times (b(q_j), b(q_{j+1})] \cap Q_s\right)$.
          \item If $a(q_{j+1}) + \sqrt{n} \leq b(q_j)$, then $I_\ell \leftarrow [q_j, q_{j+1})$, $J_\ell \leftarrow [a(q_{j+1}), b(q_j)]$ and $\ell \leftarrow \ell + 1$.
          \item Else, $E \leftarrow E \cup \left([q_j, q_{j+1}) \times [a(q_{j+1}), b(q_{j})]\right)$.
      \end{enumerate}
\end{framed}
\caption{Partitioning $Q_s$ .}\label{fig:paritioning Qs}
\end{figure}

      We start by proving that our algorithm outputs a partition.
      First, we observe that for each $s \in \{\circ\} \cup [-T, T]$ our algorithm partitions the elements of $Q_s$ between $E$ and several sets $I_\ell \times J_\ell$.
      Indeed, for $s = 0$ this is trivial. For $s \in [-T, T]$, we notice that all elements of $Q_s$ are included because $a(q_j) \leq a(q_{j+1}), b(q_j) \leq b(q_{j+1})$ and $Q_s \cap [q_j, q_{j+1}) \times [n] \subseteq [q_j, q_{j+1}) \times [a(q_j), b(q_{j+1})]$. 
      Moreover, no element is included twice because all intervals are disjoint by definition and $[q_j, q_{j+1}) \times [a(q_{j+1}), b(q_j)] \subseteq Q_s$.
      
      Next, we prove that the running time does not exceed $O(n^{3/2} \cdot T)$. For $s = \circ$ the running time is $\tilde O(n)$ as sorting suffices. For generic $s \in [-T, T]$ we prove that our algorithm runs in time $O(n^{3/2})$.
      For each $i$, computing $a(i)$ and $b(i)$ can be done in $\log n$ time through binary search. Computing the intersection $[q_j, q_{j+1}) \times [a(q_j), a(q_{j+1})) \cap Q_s$ can be done in time $\sqrt{n}\cdot |a(q_j) - a(q_{j+1})|$, which summing over $j$ gives $\sqrt{n} \cdot \sum_{j\in [\sqrt{n}]} |a(q_j) - a(q_{j+1})| = n^{3/2}$.
      Likewise, we can bound the total time used to compute the intersection $[q_j, q_{j+1}) \times (b(q_j), b(q_{j+1})) \cap Q_s$.
      
      Now, we prove that our partition satisfies the desiderata. First, we prove that $|E| = O(n^{7/4})$. 
      For $s=\circ$, we have that all sets added to $E$ are of the form $A_x \times B_x$ where either $|A_x| \le \sqrt{n}$ or $|B_x| \le \sqrt{n}$. Since all $A_x$ and $B_x$ are disjoint, we have $2n^{3/2} \ge |E \cap Q_0|$.
      For generic $s \in [-T, T]$, the proof follows the same computation that we performed to bound the running time. In steps (1) and (2) we increase $|E|$ by at most $\sqrt{n}\cdot |a(q_j) - a(q_{j+1})|$ and $\sqrt{n}\cdot |b(q_j) - b(q_{j+1})|$ respectively, and summing over $j$, we increase $|E|$ by at most $O(n^{3/2})$.
      In step (4), we are promised that $b(q_j) - a(q_{j+1}) < \sqrt n$, so we add at most $|q_{j+1} - q_j| \cdot \sqrt{n} = n$ elements to $E$. Summing over all $j \in [\sqrt{n}]$ we obtain $n^{3/2}$.      
      
      Next, we refine the sets $I_\ell$ and $J_\ell$ into smaller sets to make sure that each smaller set of pairs that we create has size $\sqrt[4]{n} \times \sqrt[4]{n}$. Notice that we have $|I_\ell|, |J_\ell| \geq \sqrt{n}$ for each $\ell$.
      Indeed, the above holds trivially for $s=\circ$, and for generic $s \in [-T, T]$ it holds because of the guard condition at step (3).
      Then, if $\sqrt[4]{n}$ divides $|J_\ell|$ it suffices to split $I_\ell$ (resp. $J_\ell$) into size-$\sqrt[4]{n}$ chunks $I_{\ell, 1} \dots I_{\ell, t}$ (resp. $J_{\ell, 1} \dots J_{\ell, t'}$) and consider all pairwise products $I_{\ell, a} \times J_{\ell, b}$ for $a \in [t]$, $b\in [t']$.
      If $\sqrt[4]{n}$ does \emph{not} divide $J_\ell$, then we have one leftover chunk $\tilde J_\ell$ of size at most $\sqrt[4]{n}$, an we add $I_\ell \times \tilde J_\ell$ to $E$.
      Since $|\tilde J_\ell| / |J_\ell| \leq 1 / \sqrt[4]{n}$, the  size of $E$ increases by at most $n^{7/4}$.
      Finally, desiderata (iii) is clearly satisfied since we construct $I_\ell \times J_\ell$ as a refinement of some $Q_s$.
      Thus, we satisfy all desiderata (i)--(iii).

    \paragraph{Two-step sampling.}
        Fix $\delta = \frac{\varepsilon^2}{\log^2 n \cdot \log^2 \Phi}$. Recall how weights are defined in \Cref{def:disc-dist}:
        \[ w_{i,j,\sigma} =  \exp(\eta \cdot (\sigma \cdot P_{ij}) \cdot D_{ij} / C_{ij}), \]
    so, whenever we sample $(i, j, \sigma)$ we can compute its weight, as we have access to $C$ and, thorough $(\balpha, \bbeta)$, also to $P$ and $D$.
      Thus, for each set $S \in \{I_\ell \times J_\ell\}_{\ell \in [L]}$, one can compute an approximation $\tilde{v}(S)$ of the volume $v(S): = \sum_{(i, j) \in S, \sigma \in \{\pm 1\}} w_{i,j,\sigma}$ using
      Theorem~\ref{lem:estimating normalization constant} so that $\tilde{v}(S) = v(S)(1 \pm\delta) $ in time
      $O(\sqrt[4]{n}/\delta) = \tilde O(\sqrt[4]{n} \cdot \poly(\eps^{-1}))$. 
      Since $L \leq n^2 / (\sqrt[4]{n} \cdot \sqrt[4]{n}) = n^{3/2}$, this can be done in total 
      time $\tilde O(n^{7/4} \cdot \poly(\eps^{-1}) )$.
      As for $E$, we can compute $v(E)$ exactly in time $|E| = O(n^{7/4})$.

      Our final sampling procedure first sample a part $S$ of the partition proportionally to $\tilde{v}(S) / \sum_{S'} \tilde{v}(S')$ and then uses $\constsampler$ to sample a $(i, j, \sigma) \in S$.
      The TV distance between this sampling 
      distribution and the distribution 
      $\lambda$ only comes from the 
      approximation of $v(S)$. Observe 
      that, assuming correctness of $\constsampler$, the total variation distance between
      the sampling distribution induced by the $v(S)$ and the distribution induced by
      the sampling with the $\tilde{v}(S)$ is
      at most $O(\delta)$, and we set $\delta = o\left( \frac{\varepsilon^2}{\log n \cdot \log^2 \Phi}\right)$. 
      
     Each call to $\constsampler$ on $I_\ell \times J_\ell$ returns $\sqrt[4]{n}$ samples, so, in order to collect $n \cdot \log^z (n)$ samples, $\arbitsample$ needs to call $\constsampler$ at most $L + n \cdot \log^z (n) / \sqrt[4]{n} = \tilde O(n^{3/2})$ times. 

      The running time of the sampling phase has two terms. The first term comes from approximating the volume of each partition, and it is $\tilde O(n^{7/4} \cdot \poly(\eps^{-1}) )$. The second term comes from running $\constsampler$. As noted above, the number of such calls is bounded by $\tilde O(n^{3/2})$ and each calls uses time $\tilde O((\sqrt[4]{n})^{2-\Omega(\phi)})$, so, the total time spent running $\constsampler$ is $\tilde O (n^{2-\Omega(\phi)} \cdot \poly(\eps^{-1}))$.
\end{proof}

\section{Designing a Consistent Rounding} \label{sec:consistent rounding}

For two sets $A,B,$ let $\mathcal{F}_{A,B} =  \{f: A \to B\}$ be the set of all functions from $A$ to $B$. Also for a set $A$, let $A[t] = \{ S \subseteq A \; | \; |S| \leq t\}$.

\begin{lemma} \label{lem:consistency lemma}

Fix any $\phi \in (0,1)$, $k,n \geq 1$, let $A$ be a set of size $n$, $\Sigma$ be an arbitrary set, $\sfR = \polylog(n)$,  and $t= \tilde{O}(\sqrt{n})$. 
For $i \in [\sfR]$, consider any fixed sequence of functions $(h_i, g_i, w_i)_{i \in [\sfR]}$ where $h_i: \Sigma \times A[t] \to \Sigma$, $g_i: \Sigma \times A[n] \to A[t]$, and $w_i: \Sigma \to \calF_{A, [k]}$. 
Fix any $\sigma_1 \in \Sigma$.
Suppose $S \subseteq A$ is a set of $m = \Theta(n^{1-\phi/2})$ uniformly sampled points 
 and consider the sequences $\sigma_1 \dots \sigma_\sfR$ and $f_1 \dots f_\sfR$ given by 
\begin{align*}
    f_i = w_i(\sigma_i) \quad \text{ and } \quad \sigma_{i+1} = h_i(\sigma_{i}, g_i(\sigma_i, S)).
\end{align*}
Define the partitions of $A$ given by $A^{(j)}_i = f_i^{-1}(j)$ for each $i \in [\sfR]$. 
 Then, with probability $1-\exp(-\Omega(\sqrt{n}))$, for each $i \in [\sfR]$,  the set $S$ $\tau$-shatters the partition $\{A^{(j)}_i\}_{j =1 \dots k}$ for $\tau = n^{1/2 + \phi}$.
\end{lemma}
\begin{proof}
First note that any function $f \in\mathcal{F}_{A,k}$ gives rise to a natural partition of $A$ via $\{f^{-1}(j)\}_{j \in [k]}$, thus in what follows refer this as the partition corresponding to $f$. 
    The critical observation is that the value of $f_{i+1} = w_i \circ h_i \circ g_i(\sigma_i, S)$ depends only on a set $Z_i = g_i(\sigma_i, S)$ of at most $t$ elements from $A$. Thus, there are only $\binom{n}{t} \leq n^t = 2^{\tilde{O}(\sqrt{n})}$ possible values for the set $Z_i$. 
    It follows that there are at most  $2^{\tilde{O}(\sqrt{n})}$ partitions which can arise from the function $f_{i}$, for any $i \in [\sfR]$, and therefore a total of $n^{tz \cdot \sfR} = 2^{\tilde{O}(\sqrt{n})}$ sequences of partitions which can arise overall by the functions $f_1,\dots,f_\sfR$, which are moreover fixed in advance and independent of $S$. Denote by $\cP$ the set of all such sequences of partitions.

    Now for any fixed partition $A = A_1 \cup \dots \cup A_k$, and fix any $i$ such that $|A_i| \geq \tau = n^{1/2+\phi}$ (if one exists), and let $X_j \in \{0,1\}$ indicate that the $j$-th sample in $S$ is contained in $A_i$. Clearly $\mathbb{E}[X_i] = \frac{{|A_i|}}{|A|}$. Thus, by Chernoff bounds, we have

\begin{equation}
    \begin{split}
           \pr{ \left|\sum_{i =1}^m X_i -  \frac{m|A_i|}{|A|}\right| \geq \frac{1}{20}\frac{m|A_i|}{|A|}}& \leq \exp(-\frac{1}{2000} \frac{m|A_i|}{|A|}) \\
           & \leq \exp(-\frac{1}{2000} n^{1/2 + \phi/2})
    \end{split}
\end{equation}
Conditioned on this event that $\left|\sum_{i =1}^m X_i -  \frac{m|A_i|}{|A|}\right| \leq \frac{1}{20}\frac{m|A_i|}{|A|}$, we have
\[ \frac{|S \cap A_i||}{|A_i|} = \frac{|S|}{|A|} \left(1 \pm \frac{1}{20}\right)\]
Thus, by a union bound over all $i \in [k]$, it follows that the sample $S$ $\tau$-shatters to the partition $A_1,\dots,A_k$ with probability at least $k\exp(-\frac{1}{2000}  n^{1/2 + \phi/2})$. We then have
\begin{equation}
    \begin{split}
        \pr{D \text{ $\tau$-shatters all partitions in all sequences in } \cP} &\geq 1- 2^{\tilde{O}(\sqrt{n})} \cdot  \exp\left(-\frac{1}{2000}  n^{1/2 + \phi/2}\right)        \\
        &  \geq 1-    \exp\left(-\frac{1}{4000}  n^{1/2 + \phi/2} \right)  
    \end{split}
\end{equation}
which completes the proof
\end{proof}

\begin{corollary} \label{cor:consistency lemma}
Fix any $\phi\in (0,1)$ and $k, n \geq 1$.  Let $A$ be a set of size $n^2$ and $\Sigma$ be an arbitrary set. Pick $S \subseteq A$ of size $n^{2-\phi/8}$ uniformly at random and let $\sfR = polylog(n)$.
For $i \in [\sfR]$, let $p_i$ be any probability distribution over $A$.
For $i \in [\sfR]$, let $Z_i \subseteq A$ the the multi-set of $\tilde O(n)$ i.i.d. samples from $p_i$.

Consider sequences of ``state'' variables $\sigma_1 \dots \sigma_\sfR \in \Sigma$ and ``paritioning functions'' $f_1 \dots f_\sfR \in \calF_{A, [k]}$.
Fix $\sigma_1 \in \Sigma$.
Suppose that, for $i \in [\sfR - 1]$: 
(i) $\sigma_{i+1}$ depends on the sample $Z_i$ and on $\sigma_{i}$; 
(ii) $p_{i}$ depends on $\sigma_i$ and $S$; 
(iii) $f_i$ depends on $\sigma_i$.
Define the partitions of $A$ given by $A^{(j)}_i = f_i^{-1}(j)$ for each $i \in [\sfR]$. 
Then, with high probability, for each $i \in [\sfR]$,  $S$ $\tau$-shatters the partition $\{A_i^{(j)}\}_{j\in B}$ with $\tau = n^{1+\phi/2}$.

\end{corollary}
\begin{proof}
    Follows immediately from Lemma \ref{lem:consistency lemma}, where $g_i$ is the process which samples the set $Z_i$ given the fixed value of $S$ and the current ``state'' $\sigma_i$; $h_i$ is the function which generates $\sigma_{i+1}$ from the samples $Z_i$ and $\sigma_i$; and $w_i$ is the function that defines the partitioning function $f_i$ given the state $\sigma_i$.
    Note that these functions are randomized, but Lemma \ref{lem:consistency lemma} holds for any possible fixing of the randomness, as it holds for any functions $g_i,h_i, w_i$. 

\end{proof}

\section{Acknowledgements}

Erik Waingarten would like to thank the National Science Foundation (NSF) which supported this research under
Grant No. CCF-2337993.

\appendix

\section{Reduction to Polynomial Aspect Ratio}\label{sec:poly-aspect}

\begin{lemma}\label{lem:bad-approx}
    There exists a randomized algorithm which runs in $O(nd + n \log n)$ time and has the following guarantees:
    \begin{itemize}
        \item \emph{\textbf{Input}}: A set $X = \{ x_1,\dots, x_n \} \subset (\R^d, \ell_1)$ and a vector $b \in V$.
        \item \emph{\textbf{Output}}: A random variable $\boldeta \in \R_{\geq 0}$.
    \end{itemize}
    With probability $1-o(1)$ over the randomness of the algorithm,
    \[\EMD_{X}(b) \leq \boldeta \leq \tilde{O}(n^2 \sqrt{d}) \cdot \EMD_X(b). \] 
\end{lemma}

\begin{proof}
    The algorithm proceeds by reducing to a one-dimensional problem in the following way:
    \begin{enumerate}
        \item Sample a random vector $\bg \sim \calN(0, I_d)$ and let $\bY = \{ \by_1,\dots, \by_n \} \subset \R$ be given by $\by_i = C\cdot \langle x_i, \bg\rangle$, for a parameter $C = O(n^2 \sqrt{d})$. Sort and re-index points so $\by_1 \leq \dots \leq \by_n$. Note that this transformation takes $O(nd) + O(n \log n)$ to project and sort.
        \item Compute $\EMD_{\bY}(b)$ in $O(n)$ time via a ``two-pointer'' algorithm. This is possible because the optimal coupling in one-dimension proceeds by greedily mapping as much supply/demand in sorted order, so one may proceed from smallest-to-largest by maintaining one pointer to current excess supply, and one to current excess demand. 
    \end{enumerate}
    It remains to prove the approximation guarantees, which follow from the fact that the map $(\R^d, \ell_1) \to (\R^1, \ell_2)$ given by the scaled projection onto $\bg$ is $\tilde{O}(n^2\sqrt{d})$-bi-Lipschitz with probability $1-o(1)$. This follows from the fact the identity $\ell_1^d \to \ell_2^d$ can only contract distances by at most $\sqrt{d}$, and the fact that with probability $1-o(1)$ over $\bg$, $1/(n^2\log n) \cdot \|x_i - x_j\|_2 \leq |\langle \bg, x_i - x_j\rangle | \leq O(\sqrt{\log n}) \|x_i - x_j\|_2$. Since we compute $\EMD_{\bY}(b)$ optimally, we obtain the desired guarantees on $\EMD_{X}(b)$.
\end{proof}

\begin{lemma} \label{lem:partition X preserve emd}
    There is a randomized algorithm with the following guarantees:
    \begin{itemize}
        \item \emph{\textbf{Input}}: A set $X = \{x_1,\dots, x_n\} \subset (\R^d, \ell_1)$, and a vector $b \in V$ with integer coordinates.
        \item \emph{\textbf{Output}}: A partition $X_1, \dots, X_t$ of $X$ which induces a partition of the vector $b$ into $b_1,\dots, b_t$.
    \end{itemize}
    With probability $0.9$, we have that each $b_1,\dots, b_t \in V$, and that
    \begin{align*}
        \EMD_X(b) = \sum_{i=1}^t \EMD_{X_i}(b_i).
    \end{align*}
    Furthermore, the maximum pairwise distance in $X_i$ is at most $\poly(nd) \cdot \EMD_X(b)$.
\end{lemma}

\begin{proof}
    The algorithm proceeds by first executing the algorithm of Lemma~\ref{lem:bad-approx} in order to obtain an approximation of $\EMD_X(b)$ in $\boldeta$. We assume the algorithm succeeds (which occurs with probability at least $1-o(1)$) so $\boldeta$ is at least $\EMD_X(b)$ and at most $\tilde{O}(n^2\sqrt{d}) \cdot \EMD_X(b)$. The partition of $X$ proceeds by imposing a randomly shifted grid of side-length $100\boldeta$, since this means that any pair $x_i, x_j$ fall in different parts with probability at most $\|x_i - x_j\|_1 / (100 \boldeta)$. In particular, we sample a random vector $\bh \sim [0,100\boldeta]^d$ and we let $X_1,\dots, X_t$ be the sets consisting of non-empty parts for each integer $a_1, \dots, a_d$,
    \[ \left\{x_i \in X :  \bh + a_{\ell} \cdot \boldeta \leq x_{i \ell} < \bh + (a_{\ell} + 1) \cdot \boldeta \right\}. \]
    Since $b \in V$ is an integer vector, the integrality of min-cost flows implies that there exists a minimizing feasible flow $\gamma \in \R^{n\times n}_{\geq 0}$ with integer entries. We let 
    \[ \bS = \sum_{i=1}^n \sum_{j=1}^n \gamma_{ij} \cdot \ind\{ x_i, x_j \text{ fall in different parts} \},\]
    and note that the randomly shifted grid construction implies
    \[ \Ex[\bS] = \sum_{i=1}^n \sum_{j=1}^n \gamma_{ij} \cdot \frac{\|x_i - x_j\|_1}{100\cdot \boldeta} \leq \frac{\EMD_X(b)}{100 \cdot \boldeta} \leq \frac{1}{100}. \]
    Since $\bS$ is a non-negative integer, it must be $0$ with probability at least $99/100$. In this case, in an optimal coupling, no coupled pairs fall in different parts. This implies $b_1,\dots, b_t \in V$, and that we may compute $\EMD_X(b)$ by summing $\EMD_{X_i}(b)$. The final claim of maximum pairwise distance follows from the diameter of each grid cell $O(d) \cdot \boldeta$, which is at most $\poly(nd) \cdot \EMD_X(b)$.
\end{proof}

The above lemma upper bounds the maximum distance, and the following claim applies a small amount of noise in order to lower bound the minimum distance while ensuring that the Earth mover's distance does not change significantly. 

\begin{claim} \label{clm:min distance}
     For any $X = \{x_1,\dots, x_n\} \subset (\R^d, \ell_1)$ and $\eps > 0$, let $\bY = \{ \by_1, \dots, \by_n \} \subset (\R^{d + d'}, \ell_1)$ with $d' = \Theta(\log n)$ where 
     \[ \by_i = (x_i, \bz_i) \in \R^{d + d'}, \]
     with $\bz_i \sim [0, \eps]^{d'}$. Then, the minimum distance $\|\by_i - \by_j\|_1 = \Omega(\eps d')$ with probability $1-o(1)$, and for any $b \in V$,
     \[ \left|\EMD_{X}(b) - \EMD_{\bY}(b) \right| \leq O(\eps n \log n). \]
\end{claim}

\begin{proof}[Proof of \Cref{lem:aspect-ratio}]
First, we invoke \Cref{lem:partition X preserve emd} and obtain a partition $X_1 \dots X_t$ and $b_1 \dots b_t$ such that each $X_i$ has diameter at most $\poly(nd) \cdot \emd_{X}(b)$ and 
\[
\emd_X(b) = \sum_{i=1}^t \emd_{X_i}(b_i).
\]
Then, we apply \Cref{clm:min distance} with precision parameter $\varepsilon' = \frac{\varepsilon}{n \log n} \cdot \emd_{X_i}(b_i)$ to each $X_i$ and, with probability $1- o(1)$, obtain $\bY_1 \dots \bY_t$ such that each $\bY_i$ has aspect ratio at most 
\[
\rho = \frac{\poly(nd) \cdot \emd_{X_i}(b_i)}{\varepsilon' \cdot \log n} \leq \poly(nd \varepsilon^{-1}).
\]
Moreover, for each $i \in [t]$ we have
\begin{equation} \label{eq:approximation of emd}
 \left|\EMD_{X_i}(b_i) - \EMD_{\bY_i}(b_i) \right| \leq O(\eps) \cdot \emd_{X_i}(b_i),
\end{equation}
which implies $|\sum_{i=1}^t \emd_{\bY_i}(b_i) - \emd_X(b)| \leq O(\varepsilon) \cdot \emd_X(b)$.
Finally, by properly rescaling (and translating) $\bY_i$ and $b_i$ we can ensure that each pair of points lies at distance at least $d \cdot \varepsilon^{-1}$, each coordinate of points in $\bY_i$ is at most $\Phi = \poly(nd\varepsilon^{-1})$, and that $\emd_{\bY_i}(b_i)$ stays constant. Then, rounding each coordinate to the closest integer in $[1, \Phi]$ perturbs each pairwise distance by at most a factor $1 \pm \varepsilon$, thus preserving \Cref{eq:approximation of emd}.

\end{proof}

\printbibliography

\end{document}